\documentclass[final]{siamltex}
\usepackage{amsmath,amssymb,graphicx}
\usepackage[all,cmtip]{xy}
\newcommand{\D}{{\mathrm{d}}}

\title{Thermodynamic Tree: The Space of Admissible Paths}

\author{Alexander N. Gorban\thanks{Department of Mathematics,
University of Leicester, UK  ({\tt ag153@le.ac.uk}).}}

\begin{document}

\maketitle

\begin{abstract}
Is a spontaneous transition from a state $x$ to a state $y$ allowed by
thermodynamics? Such a  question arises often in chemical thermodynamics and
kinetics. We ask the more formal question: is there a continuous path between
these states, along which the conservation laws hold, the concentrations remain
non-negative and the relevant thermodynamic potential $G$ (Gibbs energy, for
example) monotonically decreases? The obvious necessary condition, $G(x)\geq
G(y)$, is not sufficient, and we construct the necessary and sufficient
conditions. For example, it is impossible to overstep the equilibrium in
1-dimensional (1D) systems (with $n$ components and $n-1$ conservation laws).
The system cannot come from a state  $x$ to a state $y$ if they are on the
opposite sides of the equilibrium even if $G(x) > G(y)$. We find the general
multidimensional analogue of this 1D rule and constructively solve the problem
of the thermodynamically admissible transitions.

We study dynamical systems, which are given in a positively invariant convex
polyhedron $D$ and have a convex Lyapunov function $G$. An admissible path is a
continuous curve in $D$ along which $G$ does not increase. For $x,y \in D$,
$x\succsim y$ ($x$ precedes $y$) if there exists an admissible path from $x$ to
$y$ and $x\sim y$ if $x\succsim y$ and $y \succsim x$. The tree of $G$ in $D$
is a quotient space $D/\sim $. We provide an algorithm for the construction of
this tree. In this algorithm, the restriction of $G$ onto the 1-skeleton of $D$
(the union of edges) is used. The problem of existence of admissible paths
between states is solved constructively. The regions attainable by the
admissible paths are described.
\end{abstract}



\pagestyle{myheadings} \thispagestyle{plain} \markboth{A. N.
GORBAN}{THERMODYNAMIC TREE}

\section{Introduction}

\subsection{Motivation, ideas and a simple example}
``Applied dynamical systems" are models of real systems. The available information about
the real systems is incomplete and uncertainties of various types are encountered  in the
modeling. Often, we view them as errors: errors in the model structure, errors in
coefficients, in the state observation and many others. Nevertheless, there is an order
in this world of errors: some information is more reliable, we trust in some structures
more and even respect them as laws. Some other data are less reliable. There is an
hierarchy of reliability, our knowledge and beliefs (described, for example by R. Peierls
\cite{Peierls1980} for model making in physics). Extracting as many consequences from the
more reliable data either without or before use of the less reliable information is a
task which arises naturally.

In our paper, we study dynamical systems with a strictly convex Lyapunov
function $G$ defined in a positively invariant convex polyhedron $D$. For them,
we analyze the admissible paths, along which $G$ decreases monotonically, and
find the states that are attainable from the given initial state along the
admissible paths. The main area of applications of these systems is chemical
kinetics and thermodynamics. The motivation of our research comes from the
hierarchy of reliability of the information in these applications.

Let us discuss the motivation in more detail. In chemical kinetics, we can rank
the information in the following way. First of all, the list of reagents and
conservation laws should be known. Let the reagents be $A_1,A_2, \ldots , A_n$.
The non-negative real variable $N_i\geq 0$, the amount of $A_i$ in the mixture,
is defined for each reagent, and $N$ is the vector of composition with
coordinates $N_i$. The conservation laws are presented by the linear balance
equations:
\begin{equation}\label{balance}
b_i(N)=\sum_{j=1}^n a_i^j N_j={\rm const} \;\; (i=1, \ldots , m)\, .
\end{equation}
We assume that the linear functions $b_i(N)$ ($i=1, \ldots , m$) are
linearly independent.

The list of the components together with the balance conditions (\ref{balance}) is the
first part of the information about the kinetic model. This determines the space of
states, the polyhedron $D$ defined by the balance equations (\ref{balance}) and the
positivity inequalities $N_i\geq 0$. This is the background of kinetic models and any
further development is less reliable. The polyhedron $D$ is assumed to be bounded. This
means that there exist such coefficients $\lambda_i$ that the linear combination $\sum_i
\lambda_i b_i(N)$ has strictly positive coefficients: $\sum_i \lambda_i a_i^j >0$ for all
$j=1,\ldots, n$.

The thermodynamic functions provide us with the second level of information
about the kinetics. Thermodynamic potentials, such as the entropy, energy and
free energy are known much better than the reaction rates and, at the same
time, they give us some information about the dynamics. For example,  the
entropy increases in isolated systems. The Gibbs free energy decreases in
closed isothermal systems under constant pressure, and  the Helmholtz free
energy decreases under constant volume and temperature. Of course, knowledge of
the Lyapunov functions gives us some inequalities for vector fields of the
systems' velocity but the values of these vector fields remain unknown. If
there are some external fluxes of energy or non-equilibrium substances then the
thermodynamic potentials are not Lyapunov functions and the systems do not
relax to the thermodynamic equilibrium. Nevertheless, the inequality of
positivity of the entropy production persists and this gives us useful
information even about the open systems. Some examples are given in
\cite{obh,GorbKagan2006}.

The next, third part of the information about kinetics is the reaction
mechanism. It is presented in the form of the {\it stoichiometric equations} of
the elementary reactions:
\begin{equation}\label{stoichiometricequation}
\sum_i\alpha_{\rho i}A_i \to \sum_i \beta_{\rho i} A_i \, ,
\end{equation}
where $\rho =1, \ldots, m$ is the reaction number and the {\em stoichiometric
coefficients} $\alpha_{\rho i},\beta_{\rho i}$ ($i=1,\ldots, n$) are
nonnegative integers.

A {\em stoichiometric vector} $\gamma_{\rho}$ of the reaction
(\ref{stoichiometricequation}) is a $n$-dimensional vector with coordinates
$\gamma_{\rho i}=\beta_{\rho i}-\alpha_{\rho i}$, that is, `gain minus loss' in
the $\rho$th elementary reaction.

The concentration of $A_i$ is an intensive variable $c_i=N_i/V$, where $V>0$ is
the volume. The vector $c=N/V$ with coordinates $c_i$ is the vector of
concentrations. A non-negative intensive quantity, $r_{\rho}$, the reaction
rate, corresponds to each reaction (\ref{stoichiometricequation}). The kinetic
equations in the absence of external fluxes are
\begin{equation}\label{KinUrChem}
\frac{\D N}{\D t}=V \sum_{\rho}r_{\rho} \gamma_{\rho}\, .
\end{equation}
If the volume is not constant then equations for concentrations
include $\dot{V}$ and have different form.

For perfect systems and not so fast reactions the reaction rates are functions
of concentrations and temperature given by the {\em mass action law} and  by
the {\em generalized Arrhenius equation}.  A special relation between the
kinetic constants is given by the {\it principle of detailed balance}: For each
value of temperature $T$ there exists a positive equilibrium point where each
reaction (\ref{stoichiometricequation}) is equilibrated with its reverse
reaction. This principle  was introduced for collisions by Boltzmann in 1872
\cite{Boltzmann1872}.  Wegscheider introduced this principle for chemical
kinetics in 1901 \cite{Wegscheider1901}. Einstein in 1916 used it in the
background for his quantum theory of emission and absorption of radiation
\cite{Einstein1916}. Later, it was used by Onsager in his famous work
\cite{Onsager1931}. For a recent review see \cite{GorbanYa2011}.

At the third level of reliability of information, we select the list of components and
the balance conditions, find the thermodynamic potential, guess the reaction mechanism,
accept the principle of detailed balance and believe that we know the kinetic law of
elementary reactions. However, we still do not know the reaction rate constants.

Finally, at the fourth level of available information, we find the
reaction rate constants and can analyze and solve the kinetic
equations (\ref{KinUrChem}) or their extended version with the
inclusion of external fluxes.

Of course, this ranking of the available information is
conventional, to a certain degree. For example, some reaction rate
constants may be known even better than the list of intermediate
reagents. Nevertheless, this hierarchy of the information
availability, list of components -- thermodynamic functions --
reaction mechanism -- reaction rate constants, reflects the real
process of modelling and the stairs of available information about a
reaction kinetic system.

It seems very attractive to study the consequences of the
information of each level  separately. These consequences can be
also organized `stairwise'.  We have the hierarchy of questions: how
to find the consequences for the dynamics (i) from the list of
components, (ii) from this list of components plus the thermodynamic
functions of the mixture, and (iii) from the additional information
about the reaction mechanism.

The answer to the first question is the description of the balance polyhedron
$D$. The balance equations (\ref{balance}) together with the positivity
conditions $N_i\geq 0$ should be supplemented by the description of all the
faces. For each face, some  $N_i= 0$ and we have to specify which $N_i$ have
zero value. The list of the corresponding indices $i$, for which $N_i= 0$ on
the face, $I=\{i_1,\ldots,i_k\}$, fully characterizes the face. This problem of
{\em double description} of the convex polyhedra
\cite{Motzkin1953,Chernikova1965,Fukuda1996} is well known in linear
programming.

The list of vertices \cite{Balinski1961} and edges with the
corresponding indices is necessary for the thermodynamic
analysis. This is the {\em 1-skeleton} of $D$. Algorithms for
the construction of the 1-skeletons of balance polyhedra as
functions of the balance values were described in detail in
1980 \cite{obh}. The related problem of double description for
convex cones is very important for the pathway analysis in
systems biology \cite{Schuster2000,GagneurKlamt2004}.

In this work, we use the 1-skeleton of $D$, but the main focus is on
the second step, i.e. on the consequences of the given thermodynamic
potentials. For closed systems under classical conditions, these
potentials are the Lyapunov functions for the kinetic equations. For
example, for perfect systems we assume the mass action law.  If the
equilibrium concentrations $c^*$ are given, the system is closed and
both  temperature and volume are constant then the function
\begin{equation}
G=\sum_i c_i (\ln({c_i}/{c_i^*})-1)
\end{equation}
is the Lyapunov function; it should not increase in time. The
function $G$ is proportional to the free energy $F=RTG+{\rm const}$
(for detailed information about the Lyapunov functions for kinetic
equations under classical conditions see the textbook \cite{Yab1991}
or the recent paper \cite{Hangos2010}).

If we know the Lyapunov function $G$ then we have the necessary
conditions for the possibility of transition from the vector of
concentrations $c$ to $c'$ during the non-stationary reaction:
$G(c)\geq G(c')$ because the inequality $G(c(t_0))\geq G(c(t_0+t))$
holds for any time $t \geq 0$.

The inequality $G(c)\geq G(c')$ is necessary if we are to reach $c'$ from the initial
state $c$ by a thermodynamically admissible path, but it is not sufficient because in
addition to this inequality there are some other necessary conditions. The simplest and
most famous of them is: if $D$ is one-dimensional (a segment) then the equilibrium $c^*$
divides this segment into two parts and both $c(t_0)$ and $c(t_0+t)$ ($t>0$) are always
on the same side of the equilibrium.

\begin{figure}
\centering
\includegraphics[width=0.6\textwidth]{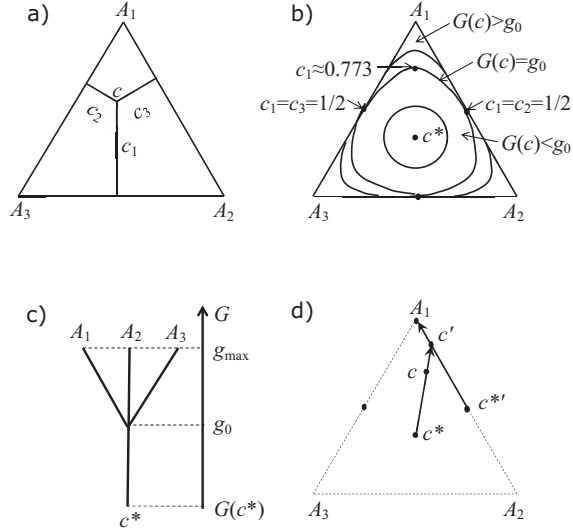}
\caption{The balance simplex (a), the levels of the Lyapunov function (b) and
the thermodynamic tree (c) for the simple system of three components, $A_1,
A_2,A_3$. Algorithm for finding a vertex $v \succsim c$ (d). \label{2Dtree}}
\end{figure}

In 1D systems the overstepping of the equilibrium is forbidden. It is impossible to
overstep a point in dimension one, but it is possible to circumvent a point in higher
dimensions. Nevertheless, in any dimension the inequality $G(c)\geq G(c')$ is not
sufficient if we are to reach $c'$ from the initial state $c$ along an admissible path.
Some additional restrictions remain in the general case as well. A two-dimensional
example is presented in Fig.~\ref{2Dtree}. Let us consider the mixture of three
components, $A_{1,2,3}$ with the only conservation law $c_1+c_2+c_3=b$ (we take for
illustration $b=1$) and the equidistribution in equilibrium $c_1^*=c_2^*=c_3^*=1/3$. The
balance polyhedron is the triangle (Fig.~\ref{2Dtree}a). In Fig.~\ref{2Dtree}b the level
sets of
$$G=\sum_{i=1}^3 c_i(\ln(3c_i) -1)$$
are presented. This function achieves its minimum at equilibrium,
$G(c^*)=-1$. On the edges, the function $G$ achieves its conditional
minimum, $g_0$, in the middles,  and $g_0=\ln(3/2)-1$. $G$ reaches
its maximal value, $g_{\max}=\ln 3 -1$, at the vertices.

If $G(c^*)< g \leq g_0$  then the level set $G(c)=g$ is connected.
If $g_0 < g \leq g_{\max}$ then the corresponding level set $G(c)=g$
consists of three components (Fig.~\ref{2Dtree}b). The critical
value is $g=g_0$. The critical level $G(c)=g_0$ consists of three
arcs. Each arc connects two middles of the edges and divides $D$ in
two sets. One of them is convex and includes two vertices, the other
includes the remaining vertex.

A thermodynamically admissible path is a continuous curve along
which $G$ does not increase. Therefore, such a path cannot intersect
these arcs `from inside', i.e. from values $G(c)\leq g_0$ to bigger
values, $G(c)> g_0$. For example, if an admissible path starts from
the state with 100\% of $A_2$, then it cannot intersect the arc that
separates the vertex with 100\% $A_1$ from two other vertices.
Therefore, any vertex cannot be reached from another one and if we
start from 100\% of $A_2$ then the reaction cannot overcome the
threshold $\sim$77.3\% of $A_1$, that is the maximum of $c_1$ on the
corresponding arc (Fig.~\ref{2Dtree}b). This is an example of the 2D
analogue of the 1D prohibition of overstepping of equilibrium.

For $x,y \in D$, $x\succsim y$ ($x$ precedes $y$) if there exists a thermodynamically
admissible path from $x$ to $y$, and $x\sim y$ if $x\succsim y$ and $y \succsim x$. The
equivalence classes with respect to $x\sim y$ in $D$ are the connected components of the
level sets $G(c)=g$. The quotient space $\mathcal{T}=D/\sim $ is the space of these
connected components. For the canonical projection we use the standard notation $\pi: D
\to \mathcal{T}$. This is the {\em tree of the connected components of the level sets} of
$G$. (Here ``tree" stands for a one dimensional continuum,  a sort of dendrites
\cite{CharatonikEncicl2003}, and not for a tree in the sense of the graph theory.)

If $x\sim y$ then $G(x)=G(y)$.  Therefore, we can define the function $G$ on the tree:
$G(\pi(c))=G(c)$. It is convenient to draw this tree on the plane with the vertical
coordinate $g=G(x)$ (Fig.~\ref{2Dtree}c). The equilibrium $c^*$ corresponds to a root of
this tree, $\pi(c^*)$. If $G(c^*)< g \leq g_0$ then the level set $G(c)=g$ corresponds to
one point on the tree. The level $G(c)= g_0$ corresponds to the branching point, and each
connected component of the level sets $G(c)=g$ with $g_0 < g \leq g_{\max}$ corresponds
to a separate point on the tree. The terminal points (``leaves" with $ g > g_0$) of the
tree correspond to the vertices of $D$.

An {\em ordered segment}  $[x,y]$ or $[y,x]$ ($x\succsim y$)  on the tree $\mathcal{T}$
consists of such points $z$ that $x \succsim z \succsim y$. A continuous curve $\varphi
:[0,1] \to D$ is an admissible path if and only if its image $\pi\circ \varphi :[0,1] \to
\mathcal{T}$ is a  path that goes monotonically down in the coordinate $g$. Such a
monotonic path in $\mathcal{T}$ from a point $x$ to the root is just a segment
$[x,\pi(c^*)]$. On this segment, each point $y$ is unambiguously characterized by
$g=G(y)$. Therefore, if for $c \in D$ we know the value $G(c)$ and a vertex $v \succsim
c$, then we can unambiguously describe the image of $c$ on the tree: $\pi (c)$ is the
point on the segment $[\pi(v),\pi(c^*)]$ with the given value of $G$, $g=G(c)$.

We can find a vertex $v \succsim c$ by a chain of central projections: the first step is
the central projection of $c$ onto the border of $D$ with center $c^*$. The result is the
point $c'$ on a face (in Fig.~\ref{2Dtree}d  this is the point $c'$ on an edge). The
second step is the central projection of the point $c'$ onto the border of the face with
the center at the partial equilibrium ${c^*}'$ (that is, the minimizer of $G$ on the
face) and so on (Fig.~\ref{2Dtree}d). If the projection on a face is the partial
equilibrium then for any vertices $v$ of the face $v \succsim c$. In particular, if the
face is a vertex $v$ then $v \succsim c$. For a simple example presented in
Fig.~\ref{2Dtree}d this is the vertex $A_1$.

In this paper, we extend these ideas and observations to any dynamical system,
which is given in a positively invariant convex polyhedron and has there a
strictly convex Lyapunov function. The class of chemical kinetic equations for
closed systems provides us standard and practically important examples of the
systems of this class.


\subsection{A bit of history}

It seems attractive to use an attainable region instead of the single
trajectory in situations with incomplete information or with information with
different levels of reliability. Such situations are typical in many areas of
science and engineering. For example, the theory for the continuous--time
Markov chain is presented in \cite{AlbertiCUZ2008,GorbanGorbanJudge2010} and
for the discrete--time Markov chains in \cite{AlbertiUhlmann1982}.

Perhaps, the first celebrated example of this approach was developed in
biological kinetics. In 1936, A.N.~Kolmogorov \cite{Kolmogorov1936} studied the
dynamics of interacting populations of prey ($x$) and predator ($y$) in the
general form:
 $$\dot{x}=xS(x,y), \;\; \dot{y}=yW(x,y)$$
under monotonicity conditions: $\partial S(x,y)/\partial y <0$, $\partial
W(x,y)/\partial y <0$. The zero isoclines, given by equations $S(x,y)=0$ or
$W(x,y)=0$, are graphs of two functions $y(x)$. These isoclines divide the
phase space into compartments with curvilinear borders. The geometry of the
intersection  of the zero isoclines, together with some monotonicity
conditions, contain important information about the system dynamics that we can
find \cite{Kolmogorov1936} without exact knowledge of the kinetic equations.
This approach to population dynamics was applied to various problems
\cite{MayLeonard1975,Bazykin1998}. The impact of this work on population
dynamics was analyzed in the review \cite{Sigmung2007}.

In 1964, Horn proposed to analyze the attainable regions for chemical reactors
\cite{Horn1964}. This approach  became popular in chemical engineering. It was
applied to the optimization of steady flow reactors \cite{Glasser1987}, to
batch reactor optimization  without knowledge of detailed kinetics
\cite{Filippi-Bossy1989}, and for optimization of the reactor structure
\cite{Hildebrandt1990}. An analysis of attainable regions is recognized as a
special geometric approach to reactor optimization \cite{Feinberg1997} and as a
crucially important part of the new paradigm of chemical engineering
\cite{Hill2009}.

Many particular applications were developed, from polymerization
\cite{SmithMalone1997} to particle breakage in a ball mill \cite{Metzger2009}
and hydraulic systems \cite{GorbKagan2006}. Mathematical methods for the study
of attainable regions vary from Pontryagin's maximum principle
\cite{McGregor1999} to linear programming \cite{Kauchali2002}, the Shrink-Wrap
algorithm \cite{Manousiouthakis2004}, and convex analysis. In 1979 it was
demonstrated how to utilize the knowledge about partial equilibria of
elementary processes to construct the attainable regions \cite{Gorban1979}. The
attainable regions significantly depend on the reaction mechanism and it is
possible to use them for the discrimination of mechanisms \cite{GorbanYa1980}.

Thermodynamic data are more robust than the reaction mechanism. Hence, there
are two types of attainable regions. The first is the thermodynamic one, which
use the linear restrictions and the thermodynamic functions
\cite{GorbanChMMS1979}. The second is generated by thermodynamics and
stoichiometric equations of elementary steps (but without reaction rates)
\cite{Gorban1979,GorbanBYa1980}. R. Shinnar and other authors
\cite{Shinnar1985} rediscovered this approach. There was even an open
discussion about priority \cite{Bykov1987}.

Some particular classes of kinetic systems have rich families of the Lyapunov
functions. Krambeck \cite{Krambeck1984} studied attainable regions for linear
systems and the $l_1$ Lyapunov norm instead of the entropy. Already simple
examples demonstrate that the sets of distributions which are accessible from a
given initial distribution by linear kinetic systems (Markov processes) with a
given equilibrium are, in general, {\em non-convex} polytopes
\cite{Gorban1979,GorbanGorbanJudge2010,Zylka1985}. The geometric approach to
attainability was developed for all the thermodynamic potentials and for open
systems as well \cite{obh}. Partial results for chemical kinetics and some
other engineering systems are summarized in \cite{Yab1991,GorbKagan2006}.

The {\em tree of the level set components} for differentiable functions was introduces in
the middle of the 20 century by Adelson-Velskii and Kronrod
\cite{AdelKronrod1945,Kronrod1950} and Reeb \cite{Reeb1946}. Sometimes these trees are
called the {\em Reeb trees} \cite{FomenkoKunii1997} but from the historical point of view
it may be better to call them the Adelson-Velskii -- Kronrod -- Reeb (or AKR) trees.
These trees were essentially used by Kolmogorov and Arnold \cite{Arnold1957} in solution
of the Hilbert's superposition problem (the ideas, their relations to dynamical systems
and role in the Arnold's scientific life are discussed in his lecture \cite{Arnold1999}).

The general {\em Reeb graph} can be defined for any topological space $X$ and real
function $f$ on it. It is the quotient space of $X$ by the equivalence relation ``$\sim$"
defined by $x \sim y$ holds if and only if $f(x) = f(y)$ and $x$, $y$ are in the same
connected component of $f^{-1}(f(x))$. Of course, this ``graph" is again not a discrete
object from the graph theory but a topological space. It has application in differential
topology (Morse theory \cite{Milnor1963}), in topological shape analysis and
visualization \cite{FomenkoKunii1997,Klemela2009}, in data analysis \cite{Stuetzle2003}
and in asymptotic analysis of fluid dynamics \cite{MaslovOmel'yanov2001,Shafarevich2000}.
The books \cite{FomenkoKunii1997,Klemela2009} include many illustration of the Reeb
graphs. The efficient mesh-based methods for the computation of the graphs of level set
components are developed for general scalar fields on 2- and 3-dimensional manifolds
\cite{Doraiswamy2009}.

Some time ago the tree of entropy in the balance polyhedra was rediscovered as an
adequate tool for representation of the attainable regions in chemical thermodynamics
\cite{GorbanChMMS1979,obh}. It was applied to analysis of various real systems
\cite{Kaganovich2010,Zarodnyuk2011}. Nevertheless, some of the mathematical backgrounds
of this approach were delayed in development and publications. Now, the thermodynamically
attainable regions are in extensive use in chemical engineering and beyond
\cite{Feinberg1997,Filippi-Bossy1989,Glasser1987,GorbKagan2006,Hildebrandt1990,Hill2009,Horn1964,Kaganovich2010,Kauchali2002,Krambeck1984,Manousiouthakis2004,McGregor1999,Metzger2009,Shinnar1988,Shinnar1985,SmithMalone1997,Zarodnyuk2011}.
In this paper we aim to provide the complete mathematical background for the analysis of
the thermodynamically attainable regions. For this purpose, we construct the trees of
strictly convex functions in a convex polyhedron. This problem allows a general meshless
solution in higher dimensions because topological and geometrical simplicity (the domain
$D$ is a convex polyhedron and the function $G$ is strictly convex in D). In this paper,
we present this solution in detail.

\subsection{The problem of attainability and its
solution\label{Sec:AttainabilityProblem}}

Let us formulate precisely the problem of attainability and its solution before
the exposition of all technical details and proofs. Our results are applicable
to any dynamical system that obeys a continuous strictly convex Lyapunov
function in a positively invariant convex polyhedron. The situations with
uncertainty, when the specific dynamical system is not given with an
appropriate accuracy but the Lyapunov function is known, give a natural area of
application of these results.

Here and below, $D$ is a convex polyhedron in $\mathbb{R}^n$, $D_0$ consists of
the vertices of $D$,  $D_1$ is the  union of the closed edges of $D$, that is,
the {\em 1-skeleton of $D$}, and $\widetilde{D_1}$ is the graph whose vertices
correspond to the vertices of $D$ and edges correspond to the edges of $D$,
(the {\em graph of the 1-skeleton}) of $D$. We use the same notations for
vertices and edges of $D$ and $\widetilde{D_1}$.

Let a real continuous function $G$ be given in  $D$. We assume that $G$ is {\em strictly
convex} in $D$ \cite{Rockafellar1970}. Let $x^*$ be the minimizer of $G$ in $D$ and let
$g^*=G(x^*)$ be the corresponding minimal value.

The {\em level set} $S_g=\{x \in D \, | \, G(x)=g\}$ is closed and the {\em sublevel set}
$U_g=\{x \in D \, | \, G(x)<g\}$ is open in $D$ (i.e. it is the intersection of an open
subset of $\mathbb{R}^n$ with $D$).

Let us transform $\widetilde{D_1}$ into a labeled graph. Each vertex $v\in D_0$
is labeled by the value $\gamma_v=G(v)$ and each edge $e=[v,w]\subset D_1$  is
labeled by the minimal value of $G$ on the segment $[v,w]\subset D$,
$g_e=\min_{[v,w]}G(x)$. The vertices and edges of $\widetilde{D_1}$ are labeled
by the same numbers as the correspondent vertices and edges of $D_1$. By
definition, the graph $\widetilde{D_1}\setminus U_g$ consists of the vertices
and edges of $\widetilde{D_1}$, whose labels $\gamma\geq g$.

The graph $\widetilde{D_1}\setminus U_{g}$ depends on $g$ but this is a
piecewise constant dependence. It changes only at $g=\gamma$, where $\gamma$
are some of the labels of the graph $\widetilde{D_1}$. Therefore, it is not
necessary to find this graph and to analyze connectivity in it for each value
$G(y)=g$.

\begin{definition}\label{Defin:AdmissiblePath}A continuous path $\varphi [0,1] \to D$ is admissible if the function
$G(\varphi (x))$ does not increase on $[0,1]$. For $x,y \in D$, $x\succsim y$ ($x$
precedes $y$) if there exists an admissible path $\varphi [0,1] \to D$ with $\varphi(0)=
x$ and $\varphi(1)= y$; $x \sim y$ if $x\succsim y$ and $y\succsim x$.
\end{definition}

The relation ``$\succsim $" is transitive. It is a {\em preorder} on $D$. The relation
``$\sim$" is an equivalence.

\begin{definition}\label{Defin:TreeofG}The tree of $G$ in $D$ is the quotient space $\mathcal{T }= D /
\sim$.
\end{definition}

The equivalence classes of $\sim$ are the path-connected components of the level sets
$S_g$. For the natural projection of $D$ on $\mathcal{T}$ we use the notation $\pi: D \to
\mathcal{T}$. We denote by $\pi^{-1}(z)\subset D$ the set of preimages of $z\in
\mathcal{T}$. The preorder ``$\succsim $" on $D$ transforms into a partial order on
$\mathcal{T}$: $\pi(x) \succsim \pi(y)$ if and only if $x\succsim y$. We call
$\mathcal{T}$ also the {\em thermodynamic tree} keeping in mind the thermodynamic
applications. The ``tree" $\mathcal{T}$ is a 1D continuum. We have to distinguish this
continuum from trees in the graph-theoretic sense which have the same graphical
representation but are discrete objects. In Sec.~\ref{Sec:Coordinates} (``Coordinates on
the thermodynamic tree") we describe the tree structure of this continuum. It includes
the root, the edges, the branching points and leaves but the edges are represented as the
real line segments.

\begin{definition}\label{Defin:OderSeg}Let $x, y \in \mathcal{T }$, $x\succsim y$. An ordered segment
$[x,y]$  (or $[y,x]$) consists of such points $z \in \mathcal{T}$ that $x \succsim z
\succsim y$.
\end{definition}

In Sec.~\ref{sec:tree} we prove that  any ordered segment $[x,y]$ ($x\neq y$) in
$\mathcal{T}$ is homeomorphic to $[0,1]$. A continuous curve $\varphi :[0,1] \to D$ is an
admissible path if and only if its image $\pi\circ \varphi :[0,1] \to \mathcal{T}$ is
monotonic in the partial order on $\mathcal{T}$. Such a monotonic path in $\mathcal{T}$
from  $x$ to $y$ ($x\succsim y$) is just a path along a segment $[x,y]$. Each point $z$
on this segment is unambiguously characterized by the value of $G(z)$.

We also use the notation $[x,y]$ for the usual closed segments in $\mathbb{R}^n$ with
ends $x,y$: $[x,y]=\{\lambda x+(1-\lambda)y \, | \, \lambda \in [0,1] \}$. The
degenerated segment $[x,x]$ is just a point $\{x\}$. The segments without one end are
$(x,y]$ and $[x,y)$ and $(x,y)$ is the segment in $\mathbb{R}^n$ without both ends.

{\em The attainability problem}: Let $x,y \in D$ and $G(x) \geq
G(y)$. Is $y$ attainable from $x$ by an admissible path?

{\em The solution of the attainability problem} can be found in
several steps:
\begin{enumerate}
\item{Find two vertices of $D$, $v_x$ and $v_y$, that
    precede $x$ and $y$, correspondingly. Such vertices
    always exist. There may be several such vertices in
    $D$. We can use any of them}.
\item{Construct the graph $\widetilde{D_1}\setminus
    U_{G(y)}$ by deletion from $\widetilde{D_1}$ all the
    elements with the labels $\gamma < G(y)$.}
\item{$y$ is attainable from $x$ by an admissible path if
    and only if $v_x$ and $v_y$ are connected
    in the graph $\widetilde{D_1}\setminus
    U_{G(y)}$.}
\end{enumerate}
So, to check the existence of an admissible path from $x$ to $y$ we should
check the inequality $G(x) \geq G(y)$ (the necessary condition) then go up in
$G$ values and find the vertices, $v_x$ and $v_y$, that precede $x$ and $y$,
correspondingly (such vertices always exist). Then we should go down in $G$
values to $G(y)$ and check whether the vertices $v_x$ and $v_y$ are connected
in the graph $\widetilde{D_1}\setminus  U_{G(y)}$. The classical problem of
determining whether two vertices in a graph are connected may be solved by many
search algorithms \cite{PearlKorf1987,NagamochiIbaraki2008}, for example, by
the elementary breadth--first or depth--first search algorithms.

The procedure ``find a vertex $v_x \in D_0$ that precedes $x\in D$" can be
implemented as follows:
\begin{enumerate}
\item{If $x=x^*$ then any vertex $v\in D_0$ precedes $x$.}
\item{If $x\neq x^*$ then consider the ray  $r_x= \{x^*+\lambda (x-x^*)\, | \,
    \lambda \geq 0 \}$. The intersection $r_x \cap D$ is a closed segment $[x^*,x']$.
    We call $x'$ the central projection of $x$ onto the border of $D$ with the center
    $x^*$; $x'\succsim  x$.}
\item{The central projection $x'$ always belongs to an interior of a face $D'$ of
    $D$, $0\leq \dim D'< \dim D$. If $\dim D'>0$ then set $x:=x'$, $D:=D'$,
    $x^*:={\rm argmin}\{G(z)\,|\, z\in D'\}$ and go to step 1.}
\item{If $\dim D'=0$ then it is a vertex  $v \succsim x$ we are looking for.}
\end{enumerate}
The dimension of the face decreases at each step, hence, after
not more than $\dim D-1$ steps we will definitely obtain the
desired vertex. A simple example is presented in
Fig.~\ref{2Dtree}d.

The information about all connected components of $\widetilde{D_1}\setminus U_{g}$ for
all values of $g$ is summarized in the tree of $G$ in $D$, $\mathcal{T}$
(Definition~\ref{Defin:TreeofG}). The tree $\mathcal{T}$ can be described as follows
(Theorem~\ref{Prop:Coordinate}): it is the space of pairs $(g,M)$, where $g\in[\min_D
G(x),\max_D G(x)]$ and $M$ is a connected component of $\widetilde{D_1}\setminus U_{g}$,
with the partial order relation: $(g,M)\succsim (g',M')$ if $g \geq g'$ and $M \subseteq
M'$. For $x,y\in D$, $x\succsim y$ if and only if $\pi(x)\succsim \pi(y)$.

The tree $\mathcal{T}$ may be constructed gradually, by descending from the
maximal value of $G$, $g=g_{\max}$ (Sec.~\ref{Sec:Algorithm}). At $g=g_{\max}$,
the graph $\widetilde{D_1}\setminus U_{g}$ consists of the isolated vertices
with the labels $\gamma=g_{\max}$ (generically, this is one vertex). Going down
in $g$, we add to $\widetilde{D_1}\setminus U_{g}$ the elements, vertices and
edges, in descending order of their labels. After adding each element we record
the changes in the connected components of $\widetilde{D_1}\setminus U_{g}$.

For each point $z \in \mathcal{T}$, $z=(g,M)$, its preimage in $D$, $\pi^{-1}(g,M)$, may
be described by the equation $G(x)=g$ supplemented by a set of linear inequalities.
Computationally, these linear inequalities can be produced by a {\em convex hull}
operation from a finite set. This finite set is described explicitly in
Sec.~\ref{Sec:AlgAttain}.

For each point $z=(g,M)$ the set of all $z'=(g',M')$ attainable by admissible paths from
$z$ has a simple description, $g' \leq g$, $M' \supseteq M$.

The tree of $G$ in $D$ provides a workbench for the analysis of various questions about
admissible paths. It allows us to reduce the $n$-dimensional problems in $D$ to some
auxiliary questions about such 1D or even discrete objects as the tree $\mathcal{T}$ and
the labeled graph $\widetilde{D_1}$. For example, we use the thermodynamic tree to solve
the following {\it problem of attainable sets}: For a given $x\in D$ describe the set of
all $y \precsim x $ by a system of inequalities. For this purpose, we find the image of
$x$ in $\mathcal{T}$, $\pi(x)$, then define the set of all points attainable by
admissible paths from $\pi(x)$ in $\mathcal{T}$ and, finally, describe the preimage of
this set in $D$ by the system of inequalities (Sec.~\ref{Sec:AlgAttain}).

\subsection{The structure of the paper}
In Sec.~\ref{sec:CutConv}, we present several auxiliary propositions from convex
geometry. We constructively describe the result of the cutting of a convex polyhedron $D$
by a convex set $U$: The description of the connected components of $D\setminus U$ is
reduced to the analysis of the 1D continuum $D_1\setminus U$, where $D_1$ is the
1-skeleton of $D$.

In Sec.~\ref{sec:tree}, we construct the tree of level set components of a
strictly convex function $G$ in the convex polyhedron $D$ and study the
properties of this tree. The main result of this section is the algorithm for
construction of this tree (Sec.~\ref{Sec:Algorithm}). This construction is
applied to the description of the attainable sets in Sec.~\ref{Sec:AlgAttain}.
These sections include some practical recipes and it is possible to read them
independently, immediately after Introduction. Several examples of the
thermodynamic trees for chemical systems are presented in
Sec.~\ref{sec:ChemKin}.

\section{Cutting of a polyhedron $D$ by a convex set $U$\label{sec:CutConv}}

\subsection{Connected components of $D\setminus U$ and of
$D_1 \setminus U$}

Let $D$ be a convex polyhedron in $\mathbb{R}^n$. We use the notations: ${\rm Aff}(D)$ is
the minimal linear manifold that includes $D$; $d=\dim {\rm Aff}(D) =\dim D$ is the
dimension of $D$; $ri( D)$ is the interior of $D$ in  ${\rm Aff}(D)$; $ r\partial( D)$ is
the border of $D$ in ${\rm Aff}(D)$.

For $P,Q \subset \mathbb{R}^n$ the Minkowski sum is $P+Q=\{x+y\, | \, x\in P, y
\in Q\}$. The convex hull (conv) and the conic hull (cone) of a set $V\subset
\mathbb{R}^n$ are:
$${\rm conv} (V) = \left\{\sum_{i=1}^q \lambda_i v_i \, \left|
\, q>0,\, v_1, \ldots , v_q \in V, \, \lambda_1, \ldots \lambda_q > 0, \,
\sum_{i=1}^q \lambda_i =1\right. \right\}\, ;$$
$${\rm cone} (V) =
\left\{\left.\sum_{i=1}^q \lambda_i v_i \, \right|  \, q \geq 0,\, v_1, \ldots
, v_q \in V, \, \lambda_1, \ldots \lambda_q > 0, \,
 \right\}\, .$$
For a set $D\subset \mathbb{R}^n$ the following two statements
are equivalent (the {\em Minkowski--Weyl theorem}):
\begin{enumerate}
\item{For some real (finite) matrix $A$ and real vector
    $b$, $D=\{x\in \mathbb{R}^n \, | Ax \leq b \}$ ;}
\item{There are finite sets of  vectors $\{v_1, \ldots ,
    v_q\} \subset \mathbb{R}^n$ and $\{r_1, \ldots r_p \}
    \subset \mathbb{R}^n$ such that
\begin{equation}\label{WeylMink}
D    ={\rm conv}\{v_1, \ldots \, v_q \} + {\rm cone}\{r_1,
    \ldots , r_p\} \, .
\end{equation}}
\end{enumerate}
Every polyhedron has two representations, of type (1) and (2),
known as (halfspace) $H$-representation and (vertex)
$V$-representation, respectively. We systematically use both
these representations. Most of the polyhedra in our paper are
bounded, therefore, for them only the convex envelope of
vertices is used  in the $V$-representation (\ref{WeylMink}).

The $k$-skeleton of $D$, $D_k$, is the union of the closed
$k$-dimensional faces of $D$:
$$D_0 \subset  D_1 \subset \ldots \subset D_d =D\, .$$
$D_0$ consists of the vertices of $D$ and  $D_1$ is a one-dimensional continuum
embedded in $\mathbb{R}^n$. We use the notation $\widetilde{D_1}$ for the graph
whose vertices correspond to the vertices of $D$ and edges correspond to the
edges of $D$, and call this graph the {\em graph of the 1-skeleton} of $D$.

Let $U$ be a convex subset of $\mathbb{R}^n$ (it may be a non-closed
set). We use $U_0$ for the set of vertices of $D$ that belong to
$U$, $U_0=U\cap D_0$,  and $U_1$ for the set of the edges of $D$
that have non-empty intersection with $U$. By default, we consider
the closed faces of $D$, hence, the intersection of an edge with $U$
either includes some internal points of the edge or consists from
one of its ends. We use the same notation $U_1$ for the set of the
corresponding edges of $\widetilde{D_1}$.

A set $W \subset P \subset \mathbb{R}^n$ is a {\em path-connected component} of $P$  if
it is its maximal path-connected subset. In this section,  we aim to describe the
path-connected components of $D\setminus U$. In particular, we prove that these
components include the same sets of vertices as the connected components of the graph
$\widetilde{D_1}\setminus U$. This graph is produced from $\widetilde{D_1}$ by deletion
of all the vertices that belong to $U_0$ and all the edges that belong to $U_1$.

\begin{lemma}\label{Lem:Vertex5.1}
Let $x \in D\setminus U$. Then there exists such a vertex $v
\in D_0$ that the closed segment $[v,x]$ does not intersect
$U$: $[v,x] \subset D\setminus U$.
\end{lemma}

\begin{proof}Let us assume the contrary: for every vertex $v\in
D_0$ there exists such $\lambda_v \in (0,1]$ that
$x+\lambda_v(v-x) \in U$. The convex polyhedron $D$ is the
convex hull of its vertices. Therefore, $x=\sum_{v \in D_0}
\kappa_v v$ for some numbers $\kappa_v \geq 0$, $v \in D_O$,
$\sum_{v\in D_0} \kappa_v=1$.

Let $$\delta_v=\frac{\kappa_v}{\lambda_v \sum_{v'\in
D_0}\frac{\kappa_{v'}}{\lambda_{v'}}}\, .$$ It is easy to check
that $\sum_{v\in D_0}\delta_v =1$ and
\begin{equation}\label{conv5.1}
x=\sum_{v\in D_0}\delta_v (x+\lambda_v(v-x))\, .
\end{equation}
According to (\ref{conv5.1}), $x$ belongs to the convex hull of
the finite set $\{x+\lambda_v(v-x) \, | \, v\in D_0 \} \subset
U$. $U$ is convex, therefore, $x\in U$ but this contradicts to
the condition $x \notin U$. Therefore, our assumption is wrong
and there exists at least one $v \in  D_0$ such that $[v,x]
\cap U=\emptyset$. \end{proof}

So, if a point from the convex polyhedron $D$ does not belong
to a convex set $U$ then it may be connected to at least one
vertex of $D$ by a segment that does not intersect $U$. Let us
demonstrate now that if two vertices of $D$ may be connected in
$D$ by a continuous path that does not intersect $U$ then these
vertices can be connected in $D_1$ by a path that is a sequence
of edges $D$, which do not intersect $U$.

\begin{lemma}\label{LemSkeleton5.2}Let $v,v' \in D_0$,  $v,v' \notin U$. Suppose that $\varphi :
[0,1] \to (D \setminus U) $ is a continuous path, $\varphi (0)=v$ and $\varphi (1)=v'$.
Then there exists such a sequence of vertices  $\{v_0, \ldots , v_l\} \subset (D
\setminus U)$ that any two successive vertices, $v_i, v_{i+1}$, are connected by an edge
$e_{i,i+1} \subset (D_1\setminus U)$.
\end{lemma}
\begin{proof}Let us, first, prove the statement: {\em the  vertices $v,v'$ belong to
one path-connected component of $D\setminus U$ if and only if they belong to
one path-connected component of $D_1 \setminus U$.}

Let us iteratively transform the path $\varphi$. On the $k$th
iteration we construct a path that connects $v$ and $v'$ in
$D_{d-k}\setminus U$, where $d=\dim D$ and $k=1, \ldots , d-1$.
We start from a transformation of path in a face of $D$.

Let $S \subset D_{j}$ be a closed $j$-dimensional face of $D$, $j\geq 2$ and let $\psi :
[0,1] \to (D_j \setminus U)$ be a continuous path, $\psi (0)=v$, $\psi (1)=v'$ and $\psi
([0,1])\cap U=\emptyset$. We will transform $\psi$ into a continuous path $\psi_S : [0,1]
\to (D_j  \setminus U)$ with the following properties: (i) $\psi_S (0)=v$, $\psi_S
(1)=v'$, (ii) $\psi_S ([0,1])\cap U=\emptyset$, (iii) $\psi_S ([0,1])\setminus S
\subseteq \psi ([0,1])\setminus S$ and (iv) $\psi_S ([0,1])\cap ri(S)=\emptyset$. The
properties (i) and (ii) are the same as for $\psi$, the property (iii) means that all the
points of $\psi_S ([0,1])$ outside  $S$ belong also to $\psi ([0,1])$ (no new points
appear outside $S$) and the property (iv) means that there are no points of $\psi_S
([0,1])$ in $ri(S)$. To construct this $\psi_S$ we consider two cases:
\begin{enumerate}
\item{$U \cap ri(S) \neq \emptyset$, i.e. there exists $y^0
    \in U \cap ri(S)$; }
\item{$U \cap ri(S) = \emptyset$.}
\end{enumerate}
In the first case, let us project any $\psi(\tau) \in ri(S)$ onto $r\partial (S)$ from
the center $y^0$. Let $y \in S$, $y\neq y^0$. There exists such a $\lambda (y)\geq 1$
that $y^0+\lambda (y)(y-y^0)\in r\partial (S)$. This function $\lambda (y)$ is continuous
in $S\setminus \{y^0\}$. The function $\lambda (y)$ can be expressed through the
Minkowski gauge functional \cite{Handbook1993} defined for a set $K$ and a point $x$:
\begin{equation*}
\begin{split}
p_K(x)=\inf\{r>0\, | \, x\in rK\}; \;\; \lambda
(y)=\left(p_{(D-y_0)}(y-y_0)\right)^{-1}\, .
\end{split}
\end{equation*}
Let us define for any $y \in ri(S)$, $y\neq y^0$ a projection $\pi_S
(y)=y^0+\lambda (y)(y-y^0)$. This projection is continuous in $S \setminus
\{y^0\}$, and $\pi_S(y)=y $ if $y \in r
\partial (S)$. It can be extended as a continuous function onto
whole $D_j\setminus \{y^0\}$:
\begin{equation*}
\pi_S (y)=\left\{
\begin{array}{ll}
\pi_S (y) &\mbox{ if } y\in S\setminus \{y^0\}\, ;\\
y & \mbox{ if }  y\in D_j \setminus S\, .
\end{array}\right.
\end{equation*}

The center $y^0 \in U$. Because of the convexity of $U$, if $y \notin U$ then
$y^0+\lambda (y-y^0) \notin U$ for any $\lambda \geq 1$. Therefore, the path
$\psi_S(t)= \pi_S(\psi(t))$ does not intersect $U$ and satisfies all the
requirements (i)-(iv).

Let us consider the second case, $U \cap ri(S) = \emptyset$. There are the
moments of the first entrance of $\psi(t)$ in $S$ and the last going of this
path out of $S$:
$$ \tau_1 =\min\{\tau \, | \, \psi (\tau) \in S\}, \;\; \tau_2 =\max \{\tau \, | \, \psi (\tau) \in
S\} \, ,$$ $0\leq \tau_1 \leq \tau_2 \leq 1$. Let $y^1 = \psi(\tau_1)$ and $y^2
= \psi(\tau_2)$. If $y^1=y^2$ then we can just delete the loop between $\tau_1$
and $\tau_2$ from the path $\psi (\tau)$ and get the path that does not enter
$ri(S)$. So, let $y^1\neq y^2$.

These points belong to $r\partial (S)$. Let $y^S \in ri(S)$ be an arbitrary
point in the relative interior of $S$ which does not belong to the segment
$[y^1,y^2]$ ($\dim S\geq 2$). The segments $[y^1,y^S]$ and $[y^2,y^S]$ do not
intersect $U$ because the following reasons: $U \cap S \subset r\partial (S)$
(may be empty), neither $y^1$ nor $y^2$ belong to $U$, and all other points of
the 3-vertex polygonal chain $[y^1,y^S,y^2]$ belong to $ri(S)$.`

Let $P(y^1,y^S,y^2)$ be a plane that includes the chain $[y^1,y^S,y^2]$. The
intersection $S \cap P(y^1,y^S,y^2)$ is a convex polygon. The convex set  $U
\cap S \cap P(y^1,y^S,y^2)$ belongs to the border of this polygon. Therefore,
it belongs to one side of it (Fig.~\ref{PlanePolygon})  (may be empty) because
convexity of the polygon and of the set $U$. The couple of points $y^1,y^2$ cut
the border of the polygon in two connected broken lines. At least one of them
does not intersect $U$ (Fig.~\ref{PlanePolygon}). Let us substitute $\psi$ on
the interval $[\tau_1,\tau_2]$ by this broken line. The new path does not
intersect $ri(S)$. Let us use for this new path the notation $\psi_S(t)$. The
path $\psi_S$ does not intersect $ri(S)$ and $U$, and all the points on them
outside $S$ are the points on the path $\psi$ for the same values of the
argument $\tau$.

\begin{figure}
\centering
\includegraphics[width=0.3\textwidth]{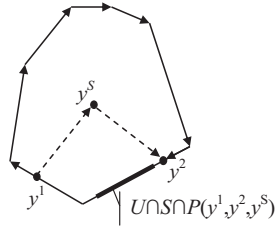}
\caption{Intersection of a face $S$ with the plane $P(y^1,y^S,y^2)$ when $U
\cap ri(S) = \emptyset$ (Lemma~\ref{LemSkeleton5.2}, case 2). In this
intersection, $U \cap S \subset r\partial (S)$ belongs to one side of the
polygon (the bold segment).  The dashed lines with arrows represent the
3-vertex polygonal chain $[y^1,y^S,y^2]$. There exists a path from $y^1$ to
$y^2$ along the boundary of the polygon. In Fig., this is the polygonal chain
that follows the solid lines with arrows. \label{PlanePolygon}}
\end{figure}

So, for any closed face  $S\subset D$ with $\dim S =j \geq 2$ and a continuous
path $\psi : [0,1] \to (D_j \setminus U)$ that connects the vertices $v,v'$ of
$D$ ($\psi (0)=v$, $\psi (1)=v'$) we construct a continuous path $\psi_S :
[0,1] \to (D_j  \setminus U)$ that connects the same vertices, does not
intersect $ri(S)$ and takes no new values outside $S$, $\psi_S ([0,1])\setminus
S \subseteq \psi ([0,1])\setminus S$.

Let us order the faces $S\subseteq D$ with $\dim S \geq 2$ in such a way that $\dim S_i
\geq \dim S_j$ for $i<j$: $D=S_0, S_1, \ldots , S_{\ell}$. Let us start from a given path
$\varphi : [0,1] \to D\setminus U$ that connects the vertices $v$ and $v'$ and let us
apply sequentially the described procedure:
$$\theta=(\ldots(((\varphi_{S_0})_{S_1})_{S_2})\ldots )_{S_{\ell}}\, .$$
By the construction, this path $\theta$ does not intersect any relative
interior $ri(S_k)$ ($k=0,1,\ldots , \ell$). Therefore, the image of $\theta$
belongs to $D_1$, $\theta : [0,1] \to (D_1 \setminus U) $. It can be
transformed into a simple path in $D_1\setminus U$ by deletion of all loops (if
they exist). This simple path (without self-intersections) is just the sequence
of edges we are looking for.
\end{proof}

Lemmas~\ref{Lem:Vertex5.1},~\ref{LemSkeleton5.2} allow us to
describe the connected components of the $d$-di\-men\-sio\-nal
set $D\setminus U$ through the connected components of the
one-dimensional continuum $D_1 \setminus U$.

\begin{proposition}\label{PropComponents5.1}Let
$W_1, \ldots, W_q$ be all the path-connected components of $D\setminus U$. Then $W_i\cap
D_0 \neq \emptyset$ for all $i=1, \ldots ,q$, the continuum $D_1\setminus U$ has $q$
path-connected components and  $W_i \cap D_1$ are these components.
\end{proposition}
\begin{proof}
Due to Lemma~\ref{Lem:Vertex5.1}, each path-connected component of $D\setminus U$
includes at least one vertex of $D$. According to Lemma~\ref{LemSkeleton5.2}, if two
vertices of $D$ belong to one path-connected component of $D\setminus U$ then they belong
to one path-connected component of $D_1 \setminus U$. The reverse statement is obvious,
because $D_1 \subset D$ and a continuous path in $D_1$ is a continuous path in $D$.
\end{proof}

We can study connected components of a simpler, discrete object, the graph
$\widetilde{D_1}$. The path-connected components of $D\setminus U$ correspond to the
connected components of the graph $\widetilde{D_1}\setminus U$. (This graph is produced
from $\widetilde{D_1}$ by deletion all the vertices that belong to $U_0$ and all the
edges that belong to $U_1$).

\begin{proposition}\label{PropComponents5.2}
Let  $W_1, \ldots, W_q$ be all the path-connected components of $D\setminus U$. Then the
graph $\widetilde{D_1}\setminus U$ has exactly $q$ connected components and each set $W_i
\cap D_0$ is the set of the vertices of $D$ of one connected component of
$\widetilde{D_1}\setminus U$.
\end{proposition}
\begin{proof}Indeed, every path between vertices in $D_1$
includes a path that connects these vertices and is the sequence of edges. (To
prove this statement we just have to delete all loops in a given path.)
Therefore, the vertices $v_1,v_2$ belong to one connected component of
$\widetilde{D_1}\setminus U$ if and only if they belong to one path-connected
component of $D_1\setminus U$. The rest of the proof follows from
Proposition~\ref{PropComponents5.1}.
\end{proof}

We proved that the path-connected components of $D\setminus U$ are in
one-to-one correspondence with the components of the graph
$\widetilde{D_1}\setminus U$ (the correspondent components have the same sets
of vertices). In applications, we will meet the following problem. Let a point
$x \in D \setminus U$ be given. Find the path-connected component of
$D\setminus U$ which includes this point. There are two basic ways to find this
component. Assume that we know the connected components of
$\widetilde{D_1}\setminus U$. First, we can examine the segments $[x,v]$ for
all vertices $v$ of $D$. At least one of them does not intersect $U$
(Lemma~\ref{Lem:Vertex5.1}). Let it be $[x,v_0]$. We can find the connected
component $\widetilde{D_1}\setminus U$ that contains $v_0$. The point $x$
belongs to the correspondent path-connected component of $D \setminus U$. This
approach exploits the $V$-description of the polyhedron $D$.  The work
necessary for this method is proportional to the number of vertices of $D$.

Another method is based on projection on the faces of $D$. Let $x \in ri(D)$.
We can take any point $y^0 \in D\setminus U$ and find the unique $\lambda_1>1$
such that $x^1=y_0+\lambda_1 (x-y^0) \in r\partial (D)$. Let $x^1 \in ri(S_1)$,
where $S_1$ is a face of $D$. If $S_1\cap U =\emptyset$ then we can take any
vertex $v_0 \in S_1$ and find the connected component $\widetilde{D_1}\setminus
U$ that contains $v_0$. This component gives us the answer. If $S_1 \cap U
\neq\emptyset$ then we can take any $y^1 \in S_1\cap U$ and find the unique
$\lambda_2>1$ such that $x^2=y^1+\lambda_2 (x^1-y^1) \in r\partial (S)$. This
$x^2$ belongs to the relative boundary of the face $S_1$. If $x^2$ is not a
vertex then it belongs to the relative interior of some face $S_2$, $\dim S_2
>0$ and we have to continue. At each iteration, the dimension of faces decreases. After $d=\dim D$
iterations at most we will get the vertex $v$ we are looking for (see also
Fig.~\ref{2Dtree}) and find the connected component of
$\widetilde{D_1}\setminus U$ which gives us the answer. Here we exploit the
$H$-description of $D$.

\subsection{Description of the connected components of $D\setminus U$ by inequalities\label{Sec:DescrIneq}}

Let $W_1, \ldots , W_q$ be the path-connected components of $D \setminus U$.

\begin{proposition}\label{Prop5.3Conv}For any set of indices $I
\subset \{1, \ldots , q\}$ the set
$$K_I= U \bigcup \left(\bigcup_{i \in I} W_i\right)$$
is convex.
\end{proposition}
\begin{proof}Let $y^1, y^2 \in K_I$. We have to prove that
$[y^1, y^2] \subset K_I$. Five different situations are
possible:
\begin{enumerate}
\item{$y^1,y^2 \in U$;}
\item{$y^1 \in U, \, y^2 \in W_i$, $i\in I$;}
\item{$y^1, y^2 \in W_i$, $i\in I$, $[y^1, y^2]\cap U
    =\emptyset$;}
\item{$y^1, y^2 \in W_i$, $i\in I$, $[y^1, y^2]\cap U \neq
    \emptyset$;}
\item{$y^1 \in W_i, \, y^2 \in W_j$, $i,j\in I$, $i\neq
    j$.}
\end{enumerate}
We will systematically use two simple facts: (i) the convexity of $U$ implies that its
intersection with any segment is a segment and (ii) if $x^1\in W_i$ and $x^2 \in
D\setminus W_i$ then the segment $[x^1,x^2]$ intersects $U$ because $W_i$ is a
path-connected component of $U$.

In case 1, $[y^1,y^2] \subset U \subset K$ because convexity
$U$.

In case 2, there exists such a point $y^3\in (y^1,y^2)$ that
$[y^1,y^3) \subseteq U \cap [y^1,y^2] \subseteq [y^1,y^3]$. The
segment $(y^3,y^2]$ cannot include any point $x \in D\setminus
W_i$ because it does not include any point from $U$. Therefore,
in this case $(y^3,y^2]\subset W_i \subset K$ and $y^3 \in K$
because it belongs either to $U$ or to $W_i$.

In case 3, $[y^1,y^2] \subset W_i \subset K$ because $W_i$ is a path-connected component
of $D\setminus U$ and $[y^1,y^2] \cap U =\emptyset$.

In case 4, $[y^1,y^2]\cap U$ is a segment $L$ with the ends $x^1,x^2$. It may
be $[x^1,x^2]$ ($y^1<x^1\leq x^2<y^2$), $(x^1,x^2]$ ($y^1\leq x^1<x^2<y^2$),
$[x^1,x^2)$ ($y^1<x^1<x^2\leq y^2$), or $(x^1,x^2)$ ($y^1\leq x^1<x^2\leq
y^2$). This segment cuts $[y^1,y^2]$ in three segments: $[y^1,y^2]=L_1\cup L
\cup L_2$, $L_1$ includes $y^1$ and $L_2$ includes $y^2$. Therefore, $L_1
\subset W_i$, $L \subset U$ and $L_2 \subset W_i$ because $W_{i}$ is a
path-connected component of $D \setminus U$ and $U$ is convex. So,
$[y^1,y^2]\subset K$.

In case 5, $[y^1,y^2]\cap U$ is also a segment $L$ with the ends $x^1,x^2$. It
may be $[x^1,x^2]$ ($y^1<x^1\leq x^2<y^2$), $(x^1,x^2]$ ($y^1\leq
x^1<x^2<y^2$), $[x^1,x^2)$ ($y^1<x^1<x^2\leq y^2$), or $(x^1,x^2)$ ($y^1\leq
x^1<x^2\leq y^2$). This segment cuts $[y^1,y^2]$ in three segments:
$[y^1,y^2]=L_1\cup L \cup L_2$, $L_1$ includes $y^1$ and $L_2$ includes $y^2$.
Therefore, $L_1 \subset W_i$, $L \subset U$ and $L_2 \subset W_j$ because
$W_{i,j}$ are path-connected components of $D \setminus U$ and $U$ is convex.
So, $[y^1,y^2]\subset K$.
\end{proof}

Typically, the set $U$ is represented by a set of inequalities, for example, $G(x) \leq
g$. It may be useful to represent the path-connected components of $D \setminus U$ by
inequalities. For this purpose, let us first construct a convex polyhedron $Q \subset U$
with the same number of path-connected components in $D \setminus Q$, $V_1, \ldots , V_q$
and with inclusons $W_i \subset V_i$. We will construct $Q$ as a convex hull of a finite
set. Let us select the edges $e$ of $D$ which intersect $U$ but the intersection $e\cap
U$ does not include vertices of $D$. For every such edge we select one point $x_e \in
e\cap U$. The set of these points is $Q_1$. By definition,
\begin{equation}\label{Qdef}
Q= {\rm conv}(U_0 \cup Q_1)\, .
\end{equation}
$Q$ is convex, hence, we can apply all the previous results about the
components of $D \setminus U$ to the components of $D \setminus Q$.

\begin{lemma}\label{Lemma5.3D}
The set $U_0 \cup Q_1$ is the set of vertices of $Q$.
\end{lemma}
\begin{proof}
A point $x \in U_0 \cup Q_1$ is not a vertex of $Q={\rm conv}(U_0 \cup Q_1)$ if and only
if it is a convex combination of other points from this set: there exist such $x_1,
\ldots, x_k \in U_0 \cup Q_1$ and $\lambda_1, \ldots , \lambda_k >0$ that $x_i \neq x$
for all $i=1, \ldots , k$ and
$$\sum_{i=1}^k \lambda_i =1\, , \;\; \sum_{i=1}^k \lambda_i x_i
=x\, .$$ If $x\in U_0$ then this is impossible because $x$ is a vertex of $D$ and $U_0
\cup Q_1 \subset D$. If $x \in Q_1$ then it belongs to the relative interior of an edge
of $D$ and, hence, may be a convex combination of points $D$ from this edge only. By
construction, $U_0 \cup Q_1$ may include only one internal point from an edge and in this
case does not include a vertex from this edge. Therefore, all the points from $Q_1$ are
vertices of $Q$.
\end{proof}

\begin{lemma}\label{Lemma5.3D-2}The set
$D\setminus Q$ has $q$ path-connected components $V_1, \ldots , V_q$ that may be
enumerated in such a way that $W_i \subset V_i$ and $W_i = V_i \setminus U$.
\end{lemma}

\begin{proof} To prove this statement about the path-connected
components, let us mention that $Q$ and $U$ include the same
vertices of $D$, the set $U_0$, and cut the same edges of $D$.
Graphs $\widetilde{D_1}\setminus Q$ and $\widetilde{D_1}\setminus U$
coincide. $Q \subset U$ because of the convexity of $U$ and
definition of $Q$. To finalize the proof, we can apply
Proposition~\ref{PropComponents5.2}.
\end{proof}

\begin{proposition}\label{Propos5.4}Let $I$ be any set of
indices from $\{1, \ldots , q\}$.
\begin{equation}\label{Eq:QandV_i}
Q\bigcup \left(\bigcup_{i \in I} V_i \right)=
{\rm conv} \left(U_0\bigcup Q_1 \bigcup \left(\bigcup_{i \in I}
\left(D_0 \bigcap V_i \right) \right)\right)
\end{equation}
\end{proposition}
\begin{proof}On the left hand side of (\ref{Eq:QandV_i}) we see
the union of $Q$ with the connected components $V_i$ ($i \in
I$). On the right hand side there is a convex envelope of a
finite set. This finite set consists of the vertices of $Q$,
($U_0\cup Q_1$) and the vertices of $D$ that belong to $V_i$
($i \in I$). Let us denote by $R_I$ the right hand side of
(\ref{Eq:QandV_i}) and by $L_I$ the left hand side of
(\ref{Eq:QandV_i}).

$L_I$ is convex due to Proposition~\ref{Prop5.3Conv} applied to $Q$ and $V_i$. The
inclusion $R_I \subseteq L_I$ is obvious because $L_I$ is convex and $R_I$ is defined as
a convex hull of a subset of $L_I$. To prove the inverse inclusion, let us consider the
path-connected components of $D \setminus R_I$. Sets $V_j$ ($j \notin I$) are the
path-connected components of $D \setminus R_I$ because they are the path-connected
components of $D \setminus Q$, $Q \subset R_I$ and $R_I \cap V_j = \emptyset$ for $j
\notin I$. There exist no other path-connected components of $Q \subset R_I$ because all
the vertices of $V_i$ ($i \in I$) belong to $R_I$ by construction, hence, $D_0 \setminus
R_I \subset \cup_{j \notin I} V_j$. Due to Lemma~\ref{Lem:Vertex5.1} every path-connected
component of $D \subset R_I$ includes at least one vertex of $D$. Therefore, $V_j$ ($j
\notin I$) are all the path-connected components of $D \setminus R_I$ and $D \setminus
R_I= \cup_{j \notin I} V_j$. Finally, $R_I = D \setminus \cup_{j \notin I} V_j = Q\cup
\left( \cup_{i \in I} V_i\right)=L_I$. \end{proof}

According to Lemma~\ref{Lemma5.3D}, each path-connected component $W_i \subset D\setminus
U$ can be represented in the form $W_i = V_i \setminus U $, where $V_i$ is a
path-connected component of $D \setminus Q$. By construction, $Q \subset U$, hence
\begin{equation}\label{cONNECTEDwVERv}
W_i=(Q\cup V_i) \setminus U\, .
\end{equation}
If $U$ is given by a system of inequalities then
representations (\ref{Eq:QandV_i}) and (\ref{cONNECTEDwVERv})
give us the possibility to represent $W_i$ by inequalities.
Indeed, the convex envelope of a finite set in
(\ref{Eq:QandV_i}) may be represented by a system of linear
inequalities. If the sets $Q\cup V_i$ and $U$ in
(\ref{cONNECTEDwVERv})  are represented by inequalities then
the difference between them is also represented by the system
of inequalities.

The description of the path-connected component of $D \setminus U$ may be constructed by
the following steps:
\begin{enumerate}
\item{Construct the graph of the 1-skeleton of $D$, this is
    $\widetilde{D_1}$;}
\item{Find the vertices of $D$ that belong to $U$, this is
    the set $U_0$;}
\item{Find the edges of $D$ that intersect $U$, this is the
    set $U_1$.}
\item{Delete from $\widetilde{D_1}$ all the vertices from
    $U_0$ and the edges from $U_1$, this is the graph
    $\widetilde{D_1}\setminus U$;}
\item{Find all the connected components of
    $\widetilde{D_1}\setminus U$. Let the sets of vertices
    of these connected components be $V_{01}, \ldots ,
    V_{0q}$;}
\item{Select the edges $e$ of $D$ which intersect $U$ but
    the intersection $e\cap U$ does not include vertices of
    $D$. For every such an edge  select one point $x_e \in
    e\cap U$. The set of these points is $Q_1$.}
\item{For every $i=1, \ldots , q$ describe the polyhedron
    $R_i={\rm conv} (U_0 \cup Q_1 \cup V_{0i})$;}
\item{There exists $ q$ path-connected components of $D \setminus U$: $W_i= R_i
    \setminus U$.}
\end{enumerate}
Every step can be performed by known algorithms including
algorithms for the solution of the double description problem
\cite{Motzkin1953,Chernikova1965,Fukuda1996} and the convex
hull algorithms \cite{PreparataShamos1985}.

\begin{figure}
\centering
\includegraphics[width=0.7\textwidth]{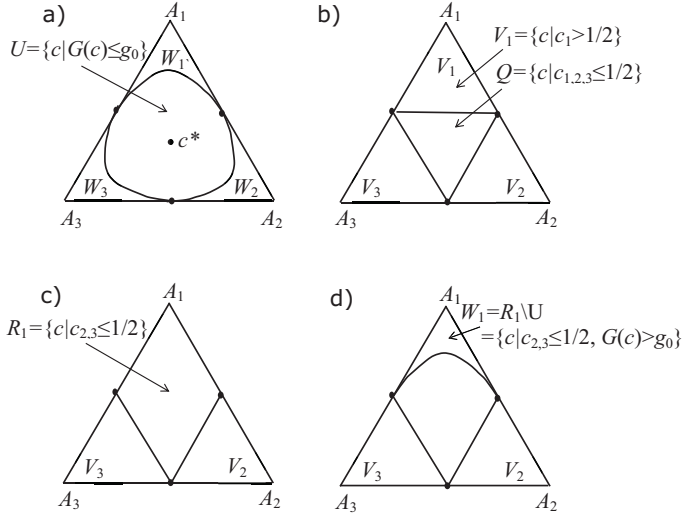}
\caption{Construction of the path-connected components $W_i$ of $D\setminus U$
for the simple example. (a) The balance simplex $D$, the set $U$ and the
path-connected components $W_i$; (b) The polyhedron $Q= {\rm conv}(U_0 \cup
Q_1)$ (\ref{Qdef}) ($U_0=\emptyset$, $Q_1$ consists of the middles of the
edges) and the connected components $V_i$ of $D\setminus Q$: $V_i=\{c\in D\, |
\, c_i>1/2\}$; (c) The set $R_1={\rm conv}(U_0\cup Q_1 \cup (D_0 \cap V_1 ) )$;
(d) The connected components $W_1$ described by the inequalities (as
$R_1\setminus U$ (\ref{Eq:QandV_i})).\label{2DConnectedComponents}}
\end{figure}

Let us use the simple system of three reagents, $A_{1,2,3}$ (Fig.~\ref{2Dtree}) to
illustrate the main steps of the construction of the path-connected components. The
polyhedron $D$ is here the 2D simplex (Fig.~\ref{2Dtree}a). The plane ${\rm Aff } D$ is
given by the balance equation $c_1+c_2+c_3=1$. We select $U=\{c\, |\, G(c)\leq g_0\}$ as
an example of a convex set (Fig.~\ref{2DConnectedComponents}a). It includes no vertices
of $D$, hence, $U_0=\emptyset$. $U$ intersects each edge of $D$ in the middle point,
hence, $U_1$ includes all the edges of $D$. The graph $\widetilde{D_1}\setminus U$
consists of three isolated vertices. Its connected components are these isolated
vertices. $Q_1$ consists of three points, the middles of the edges $(1/2,1/2,0)$,
$(1/2,0,1/2)$ and $(0,1/2,1/2)$ (in this example, the choice of these points is
unambiguous, Fig.~\ref{2DConnectedComponents}a).

The polyhedron $Q$ is a convex hull of these three points, that
is the triangle given in ${\rm Aff}(D)$ by the system of three
inequalities $c_{1,2,3}\leq 1/2$
(Fig.~\ref{2DConnectedComponents}b). The connected components
of $D \setminus Q$ are the triangles $V_i$ given in $D$ by the
inequalities $c_i >1/2$. In the whole $\mathbb{R}^3$, these
sets are given by the systems of an equation and inequalities:
$$V_i=\{ c\, | \, c_{1,2,3}\geq 0, \, c_1+c_2+c_3=1, \, c_i
>1/2\}\, .$$

The polyhedron $R_i$ is the convex hull of four points, the
middles of the edges and the $i$th vertex
(Fig.~\ref{2DConnectedComponents}c). In $D$, $R_i$ is given by
two linear inequalities, $c_j \leq 1/2, \, j\neq i$. In the
whole $\mathbb{R}^3$, these inequalities should be supplemented
by the equation and inequalities that describe $D$:
$$R_i=\{ c\, | \, c_{1,2,3}\geq 0, \, c_1+c_2+c_3=1, \, c_j
\leq 1/2\, (j\neq i)\}\, .$$

The path-connected components of $D \setminus U$, $W_i$ are described as $R_i \setminus
U$ Fig.~\ref{2DConnectedComponents}d): in $D$ we get $W_i=\{c \, | \, c_j \leq 1/2 \, (j
\neq i), \, G(c)>g_0\}$. In the whole $\mathbb{R}^3$,
$$W_i=\{ c\, | \, c_{1,2,3}\geq 0, \, c_1+c_2+c_3=1, \, c_j
\leq 1/2\, (j\neq i),\, G(c)>g_0 \}\,  .$$

$V_i$ are convex sets in this simple example, therefore, it is
possible to simplify slightly  the description of the
components $W_i$ and to represent them as $V_i\setminus U$:
$$W_i=\{ c\, | \, c_{1,2,3}\geq 0, \, c_1+c_2+c_3=1, \, c_i>1/2,\, G(c)>g_0 \}\,  $$
(or $W_i=\{c \, | \, c_i > 1/2, \, G(c)>g_0\}$ in $D$).

In the general case (more components and balance conditions),
the connected components $V_i$ may be non-convex, hence,
description of these sets by the systems of linear equations
and inequalities may be impossible. Nevertheless, there exists
another version of the description of $W_i$ where a smaller
polyhedron is used instead of $R_i$

Let $V_{0i}$ be the set of vertices of a connected component of
the graph  $\widetilde{D_1}\setminus U$. Let $E_{\rm
out}(V_{0i})$ be the set of the {\em outer edges} of $V_{0i}$
in $\widetilde{D_1}$ i.e., this is the set of edges of
$\widetilde{D_1}$ that connect vertices from $V_{0i}$ with
vertices from $D_0\setminus V_{0i}$. For each $e\in E_{\rm
out}(V_{0i})$ the corresponding edge  $e \subset D_1$
intersects $U$ because $V_{0i}$ is the set of vertices of a
connected component of the graph  $\widetilde{D_1}\setminus U$.

Let us select a point $x_e \in U \cap e$ for each $e\in E_{\rm
out}(V_{0i})$ (we use the same notations for the edges from
$\widetilde{D_1}$ and the corresponding edges from ${D_1}$).
Let us use the notation $$Q_{0i}=\{x_e \, | \, e\in E_{\rm
out}(V_{0i})\} \, .$$

\begin{proposition}\label{Prop:ShortCharact}The path-connected component $W_i$ of
$D\setminus U$ allows the following description:
$$W_i={\rm conv}(Q_{0i}\cup V_{0i})\setminus U\, .$$
\end{proposition}
\begin{proof}
The set $Q_{0i} \subset U$ and $V_{0i}$ is the set of vertices
of a connected component of the graph $\widetilde{D_1}\setminus
Q_{0i}$ by construction because $Q_{0i}$ cuts all the outer
edges of $V_{0i}$ in $D_1$. The rest of the proof follows the
proofs of Lemma~\ref{Lemma5.3D-2} and
Proposition~\ref{Propos5.4}.
\end{proof}

This proposition allows us to describe $W_i$ by the system of
inequalities. For this purpose, we have to use a convex hull
algorithm and describe the convex hull ${\rm conv}(Q_{0i}\cup
V_{0i})$ by the system of linear inequalities and then add the
inequality that describes the set $\setminus U$.

In the simple system (Fig.~\ref{2DConnectedComponents}), the connected
components of the graph $\widetilde{D_1}\setminus U$ are the isolated vertices.
The set $Q_{0i}$ for the vertex $A_i$ consists of two middles of its  incident
edges. In Fig.~\ref{2DConnectedComponents}b, the set ${\rm conv}(Q_{0i}\cup
V_{0i})$ for $V_{0i}=\{A_1\}$ is the triangle $V_1$.

\section{Thermodynamic tree\label{sec:tree}}

\subsection{Problem statement}
Let a real continuous function $G$ be given in the convex bounded polyhedron
$D\subset \mathbb{R}^n$. We assume that $G$ is {\em strictly convex} in $D$,
i.e. the set (the epigraph of $G$)
    $${\rm epi}(G)=\{(x,g)\, |\, x\in D,\, g\geq G(x)\}\subset
    D\times (- \infty,\infty)$$
is convex and for any segment $[x,y] \subset D$ ($x \neq y$) $G$ is not constant on
$[x,y]$. A strictly convex function on a bounded convex set has a unique minimizer. Let
$x^*$ be the minimizer of $G$ in $D$ and let $g^*=G(x^*)$ be the corresponding minimal
value. The level set $S_g=\{x \in D \, | \, G(x)=g\}$ is closed and the sublevel set
$U_g=\{x \in D \, | \, G(x)<g\}$ is open in $D$.  The sets $S_g$ and $D\setminus U_g$ are
compact and $S_g \subset D\setminus U_g$.

Let $x,y \in D$. According to Corollary~\ref{pathEquivalence} proven in the next
subsection, an admissible path from  $x$ to $y$ in $D$ exists if and only if $\pi(y)$
belongs to the ordered segment $[\pi(x^*),\pi(x)]$. Therefore, to describe constructively
the relation $x\succsim y$ in $D$ we have to solve the following problems:
\begin{enumerate}
\item{How to construct the thermodynamic tree $\mathcal{T}$?}
\item{How to  find an image $\pi(x)$ of a state $x\in D$ on the
    thermodynamic tree $\mathcal{T}$?}
\item{How to describe by inequalities a preimage of an ordered segment of the
    thermodynamic tree, $\pi^{-1}([w,z])\subset D$ ($w,z \in \mathcal{T}$, $z
    \succsim w $)?}
\end{enumerate}

\subsection{Coordinates on the thermodynamic tree \label{Sec:Coordinates}}

We get the following lemma directly from Definition~\ref{Defin:AdmissiblePath}. Let $x,y
\in D$.
\begin{lemma}\label{classConnected}
$x \sim y$ if and only if $G(x)=G(y)$ and $x$ and $y$ belong to the same path-connected
component of $S_g$ with $g=G(x)$.
\end{lemma}

The path-connected components of $D \setminus U_g$ can be enumerated by the
connected components of the graph $\widetilde{D_1}\setminus U_g$. The following
lemma allows us to apply this result to the path-connected components of $S_g$.

\begin{lemma}\label{Lemma5.7}Let $g>g^*$, $W_1, \ldots, W_q$ be
the path-connected components of $D \setminus U_g$ and let $\sigma_1, \ldots , \sigma_p$
be the path-connected components of $S_g$. Then $q=p$ and $\sigma_i$ may be enumerated in
such a way that $\sigma_i$ is the border of $W_i$ in $D$.
\end{lemma}
\begin{proof}
$G$ is continuous in $D$, hence, if $G(x)>g$ then there exists a
vicinity of $x$ in $D$ where $G(x)>g$. Therefore $G(y)=g$ for every
boundary point $y$ of $D\setminus U_g$ in $D$ and $S_g$ is the
boundary of $D\setminus U_g$ in $D$.

Let us define a projection $\theta_g: D\setminus U_g \to S_g$ by the
conditions: $\theta_g (x) \in [x,x^*]$ and $G(\theta_g (x))=g$. By definition,
the inequality $G(x)\geq g$ holds in $D \setminus U_g$. The function
$f_x(\lambda)=G((1-\lambda)x^*+ \lambda x)$ is strictly increasing, continuous
and convex function of $\lambda \in [0,1]$, $f_x(0) = g^*<g$, $f_x(1) =
G(x)\geq g$. The function $f_x(\lambda)$ depends continuously on $x\in
D\setminus U_g $ in the uniform metrics. Therefore, the solution $\lambda_x$ to
the equation $f_x(\lambda)=g$ on $[0,1]$ exists (the intermediate value
theorem), is unique, and continuously depends on $x \in D\setminus U_g$. The
projection $\theta_g$ is defined as $\theta_g(x)=(1-\lambda_x)x^*+ \lambda_x
x$.

The fixed points of the projection $\theta_g$ are elements of $S_g$. The image of each
path-connected component $W_i$ is a path-connected set. The preimage of every
path-connected component $\sigma_i$ is also a path-connected set. Indeed, let $\theta_g
(x) \in \sigma_i$ and $\theta_g (y) \in \sigma_i$. There exists a continuous path from
$x$ to $y$ in $D\setminus U_g$. It may be composed from three paths: (i) from $x$ to
$\theta_g (x)$ along the line segment $[x,\theta_g (x)] \subset [x,x^*]$ then a
continuous path in $\sigma_i$ between $\theta_g (x)$ and $\theta_g (y)$ (it exists
because $\sigma_i$ is a path-connected component of $S_g$ and it belongs to $D\setminus
U_g$ because $S_g \subset D\setminus U_g$) and, finally, from $\theta_g (y)$ to $y$ along
the line segment $[\theta_g (y),y] \subset [x^*,y]$. Therefore, the image of a
path-connected component $W_i$ is a path-connected components of $S_g$ that may be
enumerated by the same index $i$, $\sigma_i$. This $\sigma_i$ is the border of $W_i$ in
$D$.
\end{proof}

The equivalence class of $x \in D$ is defined as  $[x]=\{y\in D \, | \, y \sim x\}$. Let
$W(x)$ be a path-connected component of $D\setminus U_g$ ($g=G(x)$) for which $\theta_g
(W(x))=[x]$. Due to Lemma~\ref{Lemma5.7}, such a component exists and
\begin{equation}\label{Lemma5.9}W(x)=\{y \in D\, | \, y \succsim
x\}\, .
\end{equation}

Let us define a one-dimensional continuum $\mathcal{Y}$ that consists of the pairs
$(g,M)$, where $g^* \leq g \leq g_{\max}$ and $M$ is a set of vertices of a connected
component of $\widetilde{D_1}\setminus U_g$. For each $(g,M)$ the fundamental system of
neighborhoods consists of the sets $V_{\rho}$ ($\rho >0$):
\begin{equation}\label{FundamNeigh}
V_{\rho}=\{(g', M')\, | \,(g', M')\in \mathcal{Y},\, |g-g'|<\rho,\, M'\subseteq M\}\, .
\end{equation}

Let us define the partial order on $\mathcal{Y}$:
$$(g,M) \succsim (g',M') \mbox{ if } g \geq g' \mbox{ and } M \subseteq M'\, .$$
Let us introduce the mapping $\omega : D\to \mathcal{Y}$:
$$\omega (x) = (G(x),W(x)\cap D_0)\, .$$

\begin{theorem}\label{Prop:Coordinate}There exists a homeomorphism
between $\mathcal{Y}$ and $\mathcal{T}$ that preserves the partial order and makes the
following diagram commutative:
$$
\xymatrix{
D \ar[r]^{\pi} \ar[d]_{\omega} &
\mathcal{T} \ar@{<->}[ld] \\
\mathcal{Y}
}$$
\end{theorem}
\begin{proof}According to Lemmas~\ref{Lemma5.7}, \ref{Lemma5.9} and
Proposition~\ref{PropComponents5.2}, $\omega$ maps the equivalent points $x$ to the same
pair $(g,M)$ and the non-equivalent points to different pairs $(g,M)$. For any $x,y \in
D$, $x\succsim y$ if and only if $\omega (x)\succsim \omega (y)$.

The fundamental system of neighborhoods in $Y$ may be defined using this partial order.
Let us say that $(g, M)$ is compatible to $(g',M')$ if $(g', M')\succsim (g,M)$ or  $(g,
M)\succsim (g',M')\}$. Then for $\rho >0$
$$V_{\rho}=\{(g', M')\in \mathcal{Y}\, | \, |\gamma-\gamma'|<\rho \mbox{ and } (g', M') \mbox{
is compatible to } (g,M)\}\, .$$ For sufficiently small $\rho$ this definition coincides with (\ref{FundamNeigh}).

So, by the definition of $\mathcal{T}$ as a quotient space $D/ \sim$,  $\mathcal{Y}$ has
the same partial order and topology as $\mathcal{T}$. The isomorphism between
$\mathcal{Y}$ and $\mathcal{T}$ establishes one-to-one correspondence between the
$\pi$-image of the equivalence class $[x]$, $\pi([x])$, and the $\omega$-image of the
same class, $\omega([x])$.
\end{proof}

$\mathcal{Y}$ can be considered as a coordinate system on $\mathcal{T}$. Each point is
presented as a pair $(g,M)$ where $g^* \leq g \leq g_{\max}$ and $M$ is a set of vertices
of a connected component of $\widetilde{D_1}\setminus U_g$. The map $\omega$ is the
coordinate representation of the canonical projection $\pi: D \to \mathcal{T}$. Now, let
us use this coordinate system and the proof of Theorem~\ref{Prop:Coordinate} to obtain
the following corollary.

\begin{corollary}\label{pathEquivalence}
An admissible path from  $x$ to $y$ in $D$ exists if and only if
$$\pi(y) \in [\pi(x^*),\pi(x)]\, .$$
\end{corollary}
\begin{proof}
Let there exist an admissible path from  $x$ to $y$ in $D$, $\varphi : [0,1] \to D$. Then
$\pi(x) \succsim \pi (y)$ in $\mathcal{T}$. Let $\pi(x)=(G(x),M)$ in coordinates
$\mathcal{Y}$. For any $v\in M$, $\pi(y)\in [\pi(x^*),\pi(v)]$ and $\pi(x)\in
[\pi(x^*),\pi(v)]$.

Assume now that $\pi(y) \in [\pi(x^*),\pi(x)]$ and $\pi(x)=(G(x),M)$. Then  the
admissible path from $x$ to $y$ in $D$ can be constructed as follows. Let $v\in M$ be a
vertex of $D$. $G(v)\geq G(x)$ for each $v\in M$.  The straight line segment $[x^*,v]$
includes a point $x_1$ with $G(x_1)=G(x)$ and $y_1$ with $G(y_1)=G(y)$. Coordinates of
$\pi (x_1)$ and $\pi(x)$ in $\mathcal{Y}$ coincide as well as coordinates of $\pi (y_1)$
and $\pi(y)$. Therefore, $x \sim x_1$ and $y\sim y_1$. The admissible path from $x$ to
$y$ in $D$ can be constructed as a sequence of three paths: first, a continuous path from
$x$ to $x_1$ inside the path-connected component of $S_{G(x)}$
(Lemma~\ref{classConnected}), then from $x_1$ to $y_1$ along a straight line and after
that a continuous path from $y_1$ to $y$ inside the path-connected component of
$S_{G(y)}$.
\end{proof}

To describe the space $\mathcal{T}$ in coordinate representation $\mathcal{Y}$, it is
necessary to find the connected components of the graph $\widetilde{D_1}\setminus U_g$
for each $g$. First of all, this function, $$g \mapsto \mbox{ the set of connected
components of } \widetilde{D_1}\setminus U_g\, ,$$ is piecewise constant. Secondly, we do
not need to solve at each point  the computationally heavy problem of the construction of
the connected components of the graph $\widetilde{D_1}\setminus U_g$ ``from scratch". The
problem of the parametric analysis of these components as functions of $g$ appears to be
much cheaper. Let us present a solution of this problem. At the same time, this is a
method for the construction of the thermodynamic tree in coordinates $(g,M)$.

The coordinate system $\mathcal{Y}$ allows us to describe the {\em tree structure} of the
continuum $\mathcal{T}$. This structure  includes a root, $(g^*,D_0)$, edges, branching
points and leaves.

Let  $M$ be a connected component of $\widetilde{D_1}\setminus U_g$ for some $g$, $g^*
<g< g_{\max}$. If $M\subsetneqq D_0$ then the set of all points $(g,M)\in \mathcal{T}$
has for a given $M$ the form $(\underline{a}_M,\overline{a}_M] \times M$,
$\underline{a}_M<\overline{a}_M$. We call this set an {\em edge} of $\mathcal{T}$.

If $M$ includes all the vertices of $D$ ($M=D_0$) then the set of all points $(g,M)\in
\mathcal{T}$ has the form $[g^*,\overline{a}_{D_0}] \times D_0$. This may be either an
edge (if $\overline{a}_{D_0}> g^*$) or just a {\em root}, $\{(g^*,D_0)\}$, (this is
possible in 1D systems).

Let us define the numbers $\underline{a}_M=\inf\{g\, | \, (g,M)\in \mathcal{T}\}$. Let us
introduce the set of outer edges of $M$ in $\widetilde{D_1}$,  $E_{\rm out}(M)$. This is
the set of edges of $\widetilde{D_1}$ that connect vertices from $M$ with vertices from
$D_0\setminus M$. We keep the same notation, $E_{\rm out}(M)$, for the set of the
corresponding edges of $D$.
\begin{equation}\label{BranchingValue}
 \underline{a}_M=\max_{e\in E_{out}(M)} \min \{G(x) \, | \, x\in e\}\, .
\end{equation}
This number,  $\underline{a}_M$, is the ``cutting value" of $G$ for $M$. It cuts $M$ from
the other vertices of $\widetilde{D_1}$ in the following sense: if we  delete from
$\widetilde{D_1}$ all the edges $e$ with the label values $<
 \underline{a}_M$ then $M$ will remain attached to some vertices from $D_0\setminus M$.
If we delete the edges with the label values $\leq  \underline{a}_M$ then $M$ becomes
disconnected from $D_0\setminus M$. There is the only connected component of
$\widetilde{D_1}\setminus U_{\underline{a}_M}$ that includes $M$, $M' \supsetneqq M$. The
pair $(\underline{a}_M,M')\in \mathcal{T}$ is a {\em branching point} of $\mathcal{T}$.
The edge $(\underline{a}_M,\overline{a}_M] \times M$ connects two vertices, the upper
vertex $(\overline{a}_M,M)$ and the lower vertex, $(\underline{a}_M,M')$.

If $M$ consists of one vertex, $M=\{v\}$, then the point $(G(v),\{v\})$ is a {\em leaf}
of $\mathcal{T}$.

\subsection{Construction of the thermodynamic tree\label{Sec:Algorithm}}

To construct the tree of $G$ in $D$  we need the graph $\widetilde{D_1}$ of the
1-skeleton of the polyhedron $D$. Elements of $\widetilde{D_1}$ should be labeled by the
values of $G$. Each vertex $v$ is labeled by the value $\gamma_v=G(v)$ and each edge
$e=[v,w]$ is labeled by the minimal value of $G$ on the segment $[v,w]\subset D$,
$g_e=\min_{[v,w]}G(x)$. We need also the minimal value $g^*=\min_{D}\{G(x)\}$ because the
root of the tree is $(g^*,D_0)$.

The strictly convex function $G$ achieves its local maxima in $D$ only in vertices. The
vertex $v$ is a (local) maximizer of $g$ if $g_e <\gamma_v$ for each edge $e$ that
includes $v$. The leaves of the thermodynamic tree are pairs $(\gamma_v,\{v\})$ for the
vertices that are the local maximizers of $G$.

As a preliminary step of the construction, we arrange and enumerate the labels of the
elements of $\widetilde{D_1}$,  the vertices and edges, in descending order. Let there
exist $l$ different label values: $g_{\max}=a_1
> a_2 >\ldots > a_l$. Each $a_k$ is a value $\gamma_{v}=G(v)$ at a vertex $v \in D_0$
or the minimum of $G$ on an edge $e \subset D_1$ (or both). Let  $A_i$ be the
set of vertices $v\in D_0$ with $\gamma_v=a_i$ and let $E_i$ be the set of
edges of $D_1$ with $g_e=a_i$ ($i=1, \ldots , l$).

Let us construct the connected components of the graph
$\widetilde{D_1}\setminus U_g$ starting from $a_1=g_{\max}$. The function $G$
is strictly convex, hence, $a_1 =\gamma_v$ for a set of vertices $A_1 \subset
D_0$ but it is impossible that $a_1 =g_e$ for an edge $e$, hence,
$E_1=\emptyset$.

The set of connected components of $\widetilde{D_1}\setminus U_g$ is the same
for all $g\in (a_{i+1},a_i]$. For an interval $(a_2, a_1]$ the connected
components of $\widetilde{D_1}\setminus U_g$ are the one-element sets $\{v\}$
for $v \in A_1$.

For $g\in [g^*,a_l]$ the graph $\widetilde{D_1}\setminus U_g$ includes all the
vertices and edges of $\widetilde{D_1}$ and, hence, it is connected for this
segment. Let us take, formally, $a_{l+1}=g^*$.

Let $\mathcal{L}_{i}=\{M^i_1,\ldots, M_{k_i}^i\}$ be the set of the connected
components of $\widetilde{D_1}\setminus U_g$ for $g\in (a_i,a_{i-1}]$ ($i=1,
\ldots , l$). Each connected component is represented by the set of its
vertices $M^i_j$. Let us describe the recursive procedure for construction of
$\mathcal{L}_{i}$:
\begin{enumerate}
\item{Let us take formally $\mathcal{L}_{0}=\emptyset$.}
\item{Assume that $\mathcal{L}_{i-1}$ is given and $i \leq l$. Let us find the set $\mathcal{L}_{i}$
of connected components of $\widetilde{D_1}\setminus U_g$ for $g=a_i$ (and,
therefore, for $g\in(a_{i+1},a_i]$).
\begin{itemize}
\item{Add the one-element sets $\{v\}$ for all $v\in A_i$ to the set
    $$\mathcal{L}_{i-1}=\{M^{i-1}_1,\ldots, M_{k_{i-1}}^{i-1}\}\, .$$ Denote this
    auxiliary set of sets as $\widetilde{L}_{i,0}=\{M_1,\ldots, M_q\}$, where
    $q=k_{i-1}+|A_i|$.}
 \item{Enumerate the edges from $E_i$ in an arbitrary order: $e_1, \ldots,
     e_{|E_i|}$. For each $k=0, \ldots, |E_i|$, create recursively an auxiliary
set of sets $\widetilde{L}_{i,k}$ by the union of some of elements of
$\widetilde{L}_{k-1}$: Let $\widetilde{L}_{i,k-1}$ be given and $e_k$ connects
the vertices $v$ and $v'$. If $v$ and $v'$ belong to the same element of
$\widetilde{L}_{i,k-1}$ then $\widetilde{L}_{i,k}=\widetilde{L}_{i,k-1}$. If $v$
and $v'$ belong to the different elements of $\widetilde{L}_{i,k-1}$, $M$ and
$M'$, then $\widetilde{L}_{i,k}$ is produced from $\widetilde{L}_{i,k-1}$ by the
union of $M$ and $M'$:
$$\widetilde{L}_{i,k}=(\widetilde{L}_{i,k-1}\setminus \{M\}\setminus
\{M'\}) \cup \{M\cup M'\}$$ (we delete two elements, $M$ and $M'$, from
 $\widetilde{L}_{i,k-1}$  and add a new
element $M\cup M'$).}
\end{itemize}The set
$\mathcal{L}_{i}$ of connected components of $\widetilde{D_1}\setminus U_g$ for
$g=a_i$ is $\mathcal{L}_{i}=\widetilde{L}_{i,|E_i|}$.}
\end{enumerate}
Generically, all the labels of the graph $\widetilde{D_1}$ vertices and edges
are different and the sets $E_i$ and $A_i$ include not more than one element.
Moreover, for each $i$ either $E_i$ or $A_i$ is generically empty and the
description of the recursive procedure may be simplified for the generic case:
\begin{enumerate}
\item{Let us take formally $\mathcal{L}_{0}=\emptyset$.}
\item{Assume that $\mathcal{L}_{i-1}$ is given and $i \leq l$.
\begin{itemize}
\item{If $a_i$ is a label of a vertex $v$, $a_i=\gamma_v$, then add
the one-element set $\{v\}$ to the set $\mathcal{L}_{i-1}$:
$\mathcal{L}_{i}=\mathcal{L}_{i-1}\cup \{\{v\}\}$.}
 \item{Let $a_i$ be a label of an edge $e=[v,v']$.
If $v$ and $v'$ belong to the same element of $\mathcal{L}_{i-1}$ then
$\mathcal{L}_{i}=\mathcal{L}_{i-1}$.
 If $v$ and $v'$ belong to the different elements of $\mathcal{L}_{i-1}$, $M$ and
$M'$, then $\mathcal{L}_{i}$ is produced from $\mathcal{L}_{i-1}$ by the union of
$M$ and $M'$ (delete elements $M$ and $M'$ and add an element $M\cup M'$).}
\end{itemize}}
\end{enumerate}

The described procedure gives us the sets of connected components of
$\widetilde{D_1}\setminus U_g$ for all $g$ and, therefore, we get the tree $\mathcal{T}$.
The descent from the higher values of $G$ allows us to avoid the solution of the
computationally more expensive problem of the calculation of the connected components of
a graph at any level of $G$.

\subsection{The problem of attainable sets\label{Sec:AlgAttain}}

In this section, we demonstrate how to solve the problem of attainable sets. For given
$x\in D$ (an initial state) we describe the {\em attainable set} $$Att(x)=\{y\in D \, |
\, x \succsim y \}$$ by a system of inequalities. Let the tree $\mathcal{T}$ of $G$ in
$D$ be given and let all the pairs $(g,M)\in \mathcal{T}$ be described. We also use the
notation $Att(z)$ for sets attainable in $\mathcal{T}$ from $z\in \mathcal{T}$.

First of all, let us describe the {\em preimage of a point $(g,M)\in \mathcal{T}$ in
$D$}. It can be described by the equation $G(x)=g$ and a set of linear inequalities. For
each edge $e$ we select a minimizer of $G$ on $e$, $x_e={\rm argmin}\{G(x)\, | \, x \in
e\}$ (we use the same notations for the elements of the graph $\widetilde{D_1}$ and of
the continuum $D_1$). Let
$$Q_M=\{ x_e \, | \, e\in E_{\rm out}(M)\}\, .$$
In particular, $ \underline{a}_M=\max \{G(x) \, | \, x\in Q_M \}$.

The following Proposition is a direct consequence of Proposition~\ref{Prop:ShortCharact}.

\begin{proposition}The preimage of $(g,M)$ in $D$ is a set
\begin{equation}\label{Eq:Preimage}
\pi^{-1}(g,M)=\{x\in {\rm conv} (Q_M\cup M) \, | \, G(x)=g \}\, .
\end{equation}
\end{proposition}

The sets $M$ and $Q_M$ in (\ref{Eq:Preimage}) do not depend on the specific
value of $g$. It is sufficient that the point $(g,M)\in \mathcal{T}$ exists.

Let us consider the second projection of $\mathcal{T}$, i.e., the set of all connected
components of the graph $\widetilde{D_1}\setminus U_g$ for all $g$. For a connected
component $M$, the {\em lower chain} of connected components is a sequence $M=M_1
\subsetneqq M_2 \subsetneqq \ldots \subsetneqq M_k$. (``Lower" here means the descent in
the natural order in $\mathcal{T}$, $\succsim$.) For a given initial element $M=M_1$ the
{\em maximal lower chain } of $M$ is the lower chain of $M$ that cannot be extended by
adding new elements. By construction of connected components, the maximal lower chain of
$M$ is unique for each initial element $M$. In the maximal lower chain
$\underline{a}_{M_i}=\overline{a}_{M_{i+1}}$.

For each set of values $H \subset (\underline{a}_M,\overline{a}_M]$ the
preimage of the set $H\times M \subset \mathcal{T}$ is given by
(\ref{Eq:Preimage}) as
\begin{equation}\label{Eq:PreimageSet}
\pi^{-1}(H\times M)=\{x\in {\rm conv} (Q_M\cup M) \, | \, G(x)\in H \}\, .
\end{equation}

We describe the set $Att(x)$ for $x \in D$ by the following
procedures: (i) find the projection $\pi(x)$ of $x$ onto
$\mathcal{T}$, (ii) find the attainable set in $\mathcal{T}$
from $\pi(x)$, $Att(\pi(x))$, and (iii) find the preimage of
this set in $D$:
\begin{equation}\label{PreimageAtt}
Att(x)= \pi^{-1}(Att(\pi(x)))\, .
\end{equation}

The attainable set $Att(g,M)$ in $\mathcal{T}$ from $(g,M) \in
\mathcal{T}$ is constructed as a union of edges and its parts.
Let $M=M_1 \subsetneqq M_2 \subsetneqq \ldots \subsetneqq
M_k=D_0$ be the maximal lower chain of $M$. Then
\begin{equation}\label{Att()}
\begin{split}
Att(g,M)=&(\underline{a}_1,g] \times M_1 \cup (\underline{a}_2,\underline{a}_1]
\times M_2 \cup \ldots \\ &\cup (\underline{a}_{k-1},\underline{a}_{k-2}]
\times M_{k-1} \cup [\underline{a}_{k},\underline{a}_{k-1}] \times M_k\, ,
\end{split}
\end{equation}
 where $\underline{a}_i = \underline{a}_{M_i}$.

To find the preimage of $Att(g,M)$ in $D$ we have to apply
formula (\ref{Eq:PreimageSet}) to each term of (\ref{Att()}).
In Sec.~\ref{Sec:AttainabilityProblem} we demonstrated how to
find $\pi(x)$. Therefore, each step of the solution of the
problem of attainable set (\ref{PreimageAtt}) is presented.

\section{Chemical thermodynamics: examples\label{sec:ChemKin}}
\begin{figure}
\centering
\includegraphics[width=0.8\textwidth]{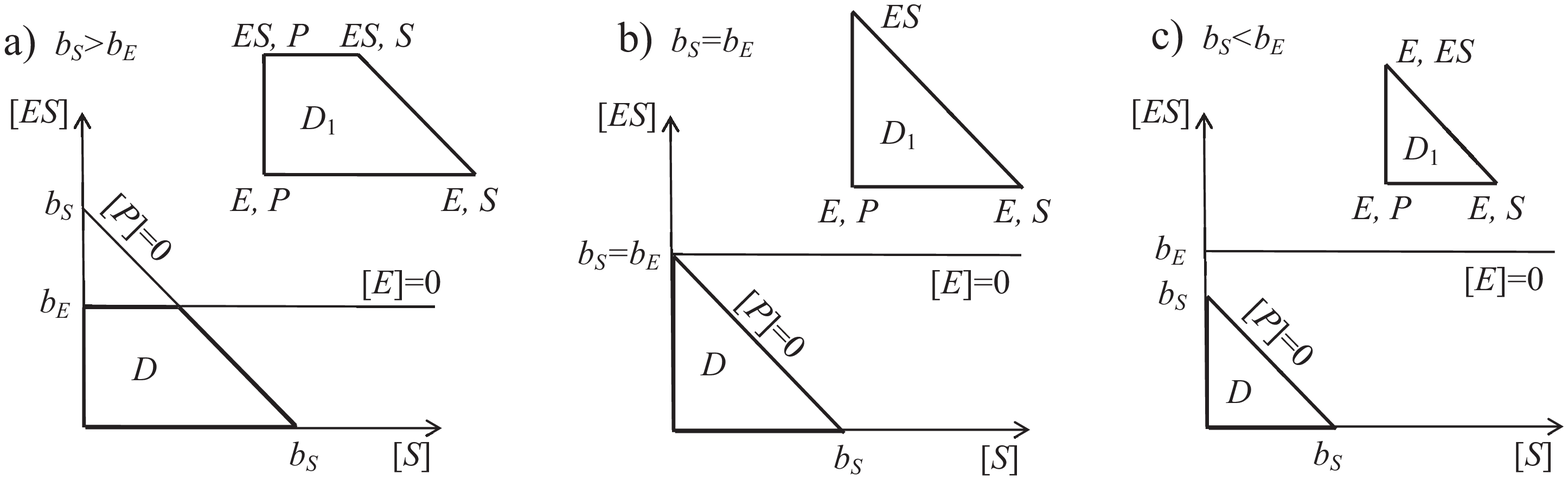}
\caption{\label{Trapecia} The balance polygon $D$ on the plane with coordinates
$[S]$ and $[ES]$ for the four--component enzyme--substrate system $S$, $E$,
$ES$ $P$ with two balance conditions, $b_S=[S]+[ES]+[P]={\rm const}$ and
$b_E=[E]+[ES]={\rm const}$.}
\end{figure}
\begin{figure}
\centering
\includegraphics[width=0.8\textwidth]{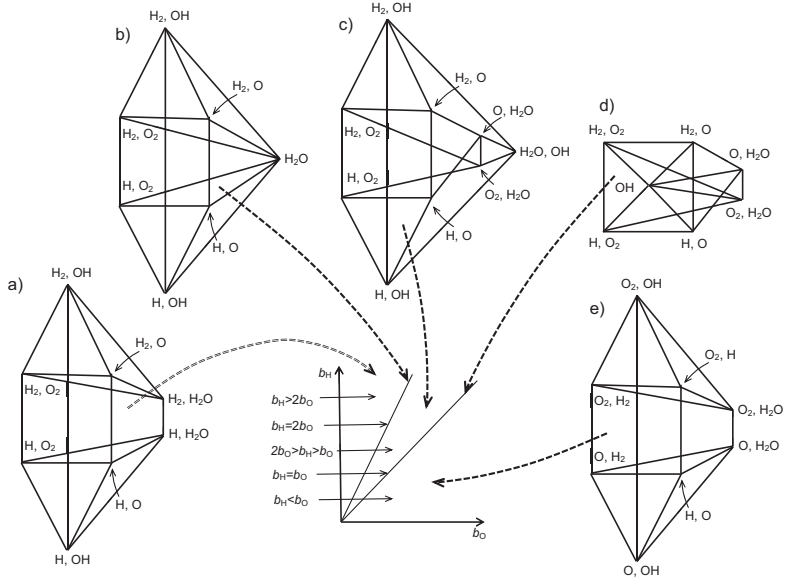}
\caption{\label{Paramet6} The graph $\widetilde{D_1}(b)$ of the one-skeleton of
the balance polyhedron for the six-component system, ${\rm H}_2$, ${\rm O}_2$,
${\rm H}$, ${\rm O}$, ${\rm H}_2{\rm O}$, ${\rm OH}$, as a piece-wise constant
function of $b=(b_{\rm H}, b_{\rm O})$. For each vertex the components are
indicated which have non-zero concentrations at this vertex.}
\end{figure}

\subsection{Skeletons of the balance polyhedra}
In chemical thermodynamics and kinetics, the variable $N_i$ is the amount of
the $i$th component in the system. The balance polyhedron $D$ is described by
the positivity conditions $N_i \geq 0$ and the balance conditions
(\ref{balance}) $b_i(N)={\rm const}$ ($i=1, \ldots , m$). Under the isochoric
(the constant volume) conditions, the concentrations $c_i$ also satisfy the
balance conditions and we can construct the balance polyhedron for
concentrations. Sometimes, the balance polyhedron is called the {\it reaction
simplex} with some abuse of language because it is not obligatory a simplex
when the number $m$ of the independent balance conditions is greater than one.

The graph $\widetilde{D_1}$ depends on the values of the balance functionals
$b_i=b_i(N)=\sum_{j=1}^n a_i^j N_j$. For the positive vectors $N$, the vectors
$b$ with coordinates $b_i=b_i(N)$ form a convex polyhedral cone in
$\mathbb{R}^m$. Let us denote this cone by $\Lambda$. $\widetilde{D_1}(b)$ is a
piece-wise constant function on $\Lambda$. Sets with various constant values of
this function are cones. They form a partition of $\Lambda$. Analysis of this
partition and the corresponding values of $\widetilde{D_1}$ can be done by the
tools of linear programming \cite{obh}. Let us represent several examples.

In the first example, the reaction system consists of four
components: the substrate $S$, the enzyme $E$, the enzyme-substrate
complex $ES$ and the product $P$. we consider the system under
constant volume. We denote the concentrations by $[S]$, $[E]$,
$[ES]$ and $[P]$. There are two balance conditions:
$b_S=[S]+[ES]+[P]={\rm const}$ and $b_E=[E]+[ES]={\rm const}$.

\begin{figure}
\centering
\includegraphics[width=0.8\textwidth]{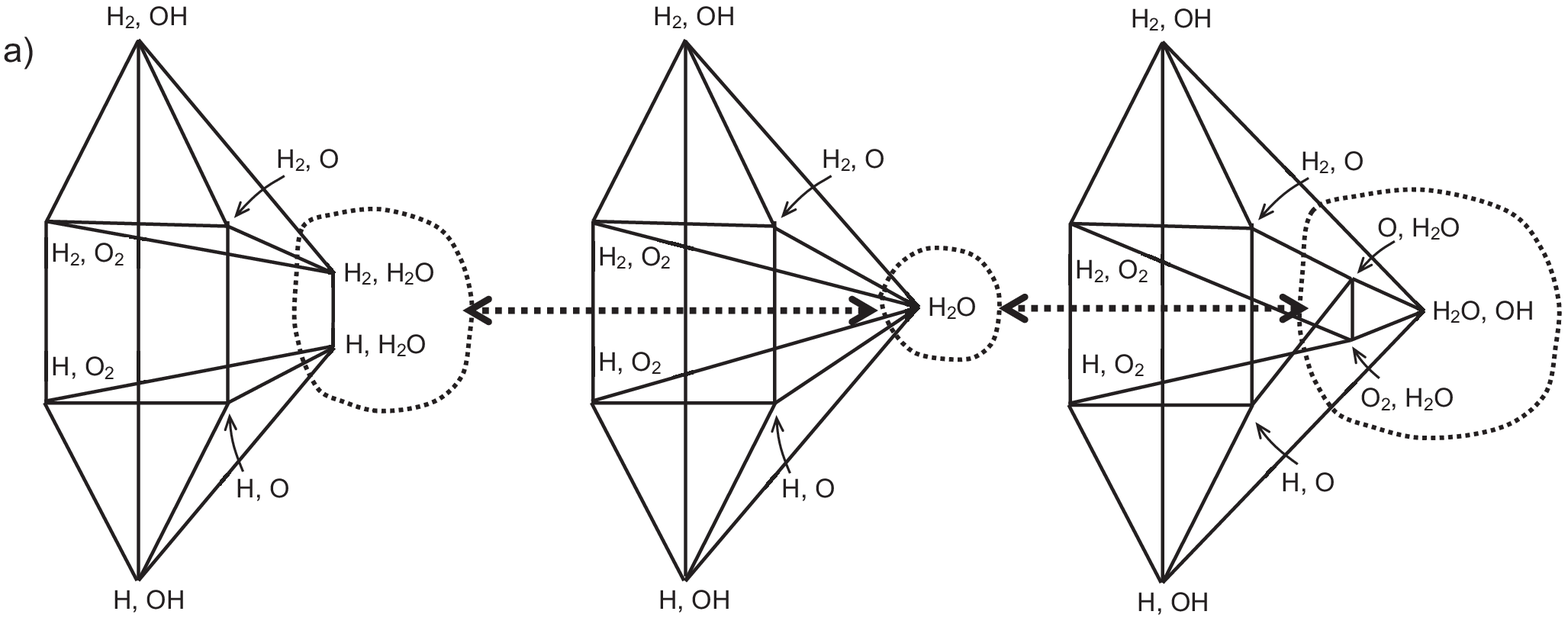}
\includegraphics[width=0.8\textwidth]{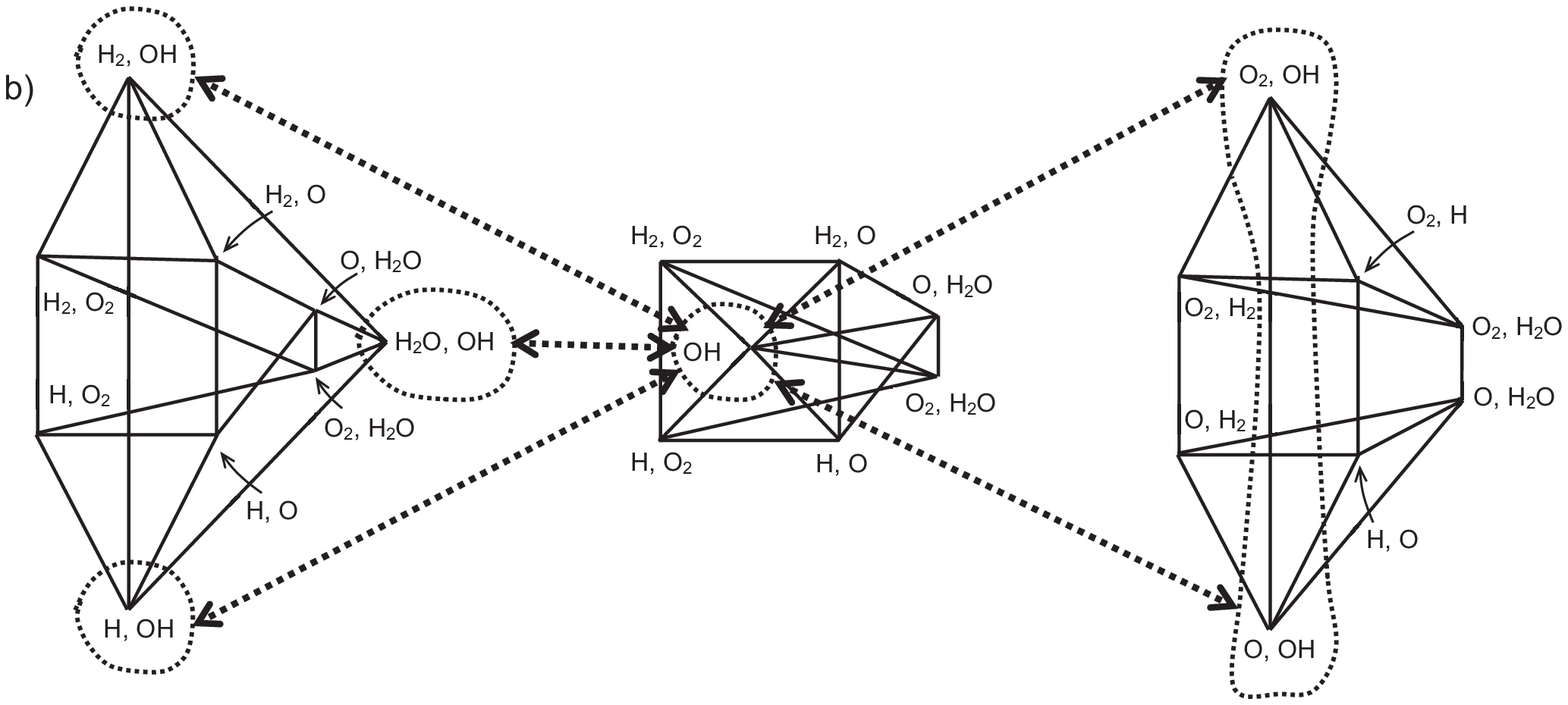}
\caption{\label{Bifur6}Transformations of the graph
$\widetilde{D_1}(b)$ with changes of the relation between $b_{\rm
H}$ and $b_{\rm O}$: (a) transition from the regular case $b_{\rm
H}> 2b_{\rm O}$ to the regular case $2b_{\rm O}>b_{\rm H}>b_{\rm O}$
through the singular case $b_{\rm H}= 2b_{\rm O}$, (b) transition
from the regular case $2b_{\rm O}>b_{\rm H}>b_{\rm O}$ to the
regular case  $b_{\rm H}<b_{\rm O}$ through the singular case
$b_{\rm H}=b_{\rm O}$.}
\end{figure}

\begin{figure}[t]
\centering
\includegraphics[width=0.4\textwidth]{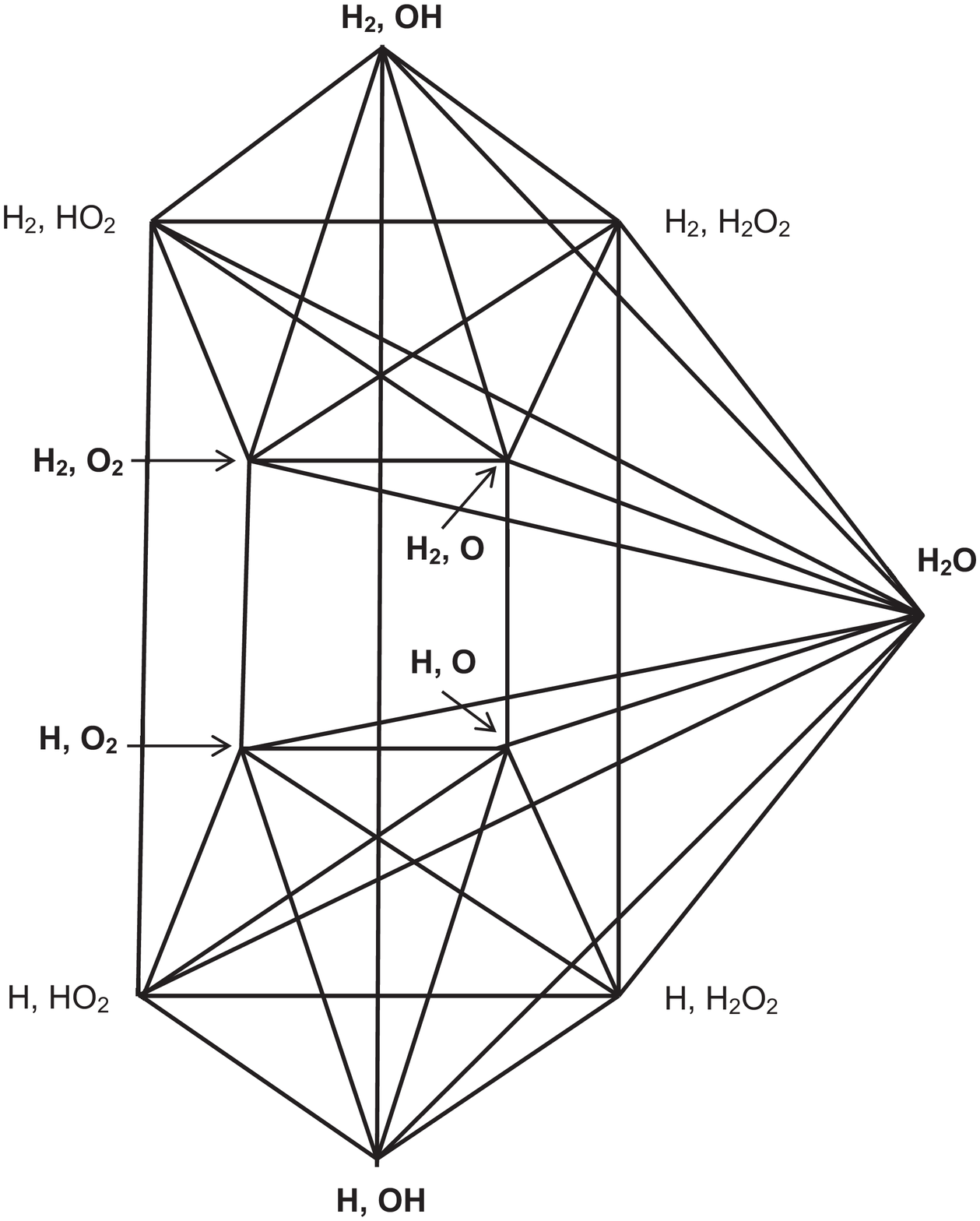}
\caption{\label{Poly8}The graph $\widetilde{D_1}$ for the
eight-component system, ${\rm H}_2$, ${\rm O}_2$, ${\rm H}$, ${\rm
O}$, ${\rm H}_2{\rm O}$, ${\rm OH}$, ${\rm H}_2{\rm O}_2$, ${\rm
H}{\rm O}_2$ for the stoichiometric mixture, $b_{\rm H}=2b_{\rm O}$.
The vertices that correspond also to the six-component mixture are
distinguished by bold font.}
\end{figure}

For  $b_S>b_E$ the polyhedron (here the polygon) $D$ is a trapezium
(Fig.~\ref{Trapecia}a). Each vertex corresponds to two components that have
non-zero concentrations in this vertex. For $b_S>b_E$ there are four such
pairs, $(ES,P)$, $(ES,S)$, $(E,P)$ and $(E,S)$. For two pairs there are no
vertices: for $(S,P)$ the value $b_E$ is zero and for $(ES,E)$ it should be
$b_S<b_E$. When $b_S=b_E$, two vertices, $(ES,P)$ and $(ES,S)$, transform into
one vertex with one non-zero component, $ES$, an the polygon $D$ becomes a
triangle (Fig.~\ref{Trapecia}b). When $b_S<b_E$ then $D$ is also a triangle and
a vertex $ES$ transforms in this case into $(ES,E)$ (Fig.~\ref{Trapecia}c).

For the second example, we select a system with six components and two balance
conditions: ${\rm H}_2$, ${\rm O}_2$, ${\rm H}$, ${\rm O}$, ${\rm H}_2{\rm O}$,
${\rm OH}$;
\begin{equation*}
\begin{split}
&b_{\rm H}=2N_{\rm H_2}+N_{\rm H}+2N_{\rm H_2O}+N_{\rm OH}\, , \\
&b_{\rm O}=2N_{\rm O_2}+N_{\rm O}+N_{\rm H_2O}+N_{\rm OH}\, .
\end{split}
\end{equation*}

The cone $\Lambda$ is a positive quadrant on the plane with the coordinates
$b_{\rm H}, b_{\rm O}$. The graph $\widetilde{D_1}(b)$ is constant in the
following cones in $\Lambda$ ($b_{\rm H}, b_{\rm O}>0$): (a) $b_{\rm H}>
2b_{\rm O}$, (b) $b_{\rm H}=2b_{\rm O}$, (c) $2b_{\rm O}>b_{\rm H}>b_{\rm O}$,
(d) $b_{\rm H}=b_{\rm O}$ and (e) $b_{\rm H}<b_{\rm O}$ (Fig.~\ref{Paramet6}).

The cases (a) $b_{\rm H}> 2b_{\rm O}$, (c) $2b_{\rm O}>b_{\rm H}>b_{\rm O}$,
and (e) $b_{\rm H}<b_{\rm O}$ (Fig.~\ref{Paramet6}) are regular: there are two
independent balance conditions and for each vertex there are exactly two
components with non-zero concentration. In case (a) ($b_{\rm H}> 2b_{\rm O}$),
if  $b_{\rm H}\to  2b_{\rm O}$ then two regular vertices, ${\rm H}_2, \,{\rm
H}_2{\rm O}$ and ${\rm H}, \,{\rm H}_2{\rm O}$, join in one vertex (case (b))
with only one non-zero concentration, ${\rm H}_2{\rm O}$ (Fig.~\ref{Bifur6}a).
This vertex explodes in three vertices ${\rm O}, \,{\rm H}_2{\rm O}$; ${\rm
O}_2, \,{\rm H}_2{\rm O}$ and ${\rm H}_2{\rm O}, \, {\rm OH}$, when $b_{\rm H}$
becomes smaller than $2b_{\rm O}$ (case (c), $2b_{\rm O}>b_{\rm H}>b_{\rm O}$)
(Fig.~\ref{Bifur6}a). Analogously, in the transition from the regular case (c)
to the regular case (e) through the singular case (d) ($b_{\rm H}=b_{\rm O}$)
three vertices join in one, ${\rm 0H}$ that explodes in two
(Fig.~\ref{Bifur6}b).

For the modeling of hydrogen combustion, the eight-component model
is used usually: ${\rm H}_2$, ${\rm O}_2$, ${\rm H}$, ${\rm O}$,
${\rm H}_2{\rm O}$, ${\rm OH}$, ${\rm H}_2{\rm O}_2$, ${\rm H}{\rm
O}_2$. In Fig.~\ref{Poly8} the graph $\widetilde{D_1}$ is presented
for one particular relations between $b_{\rm H}$ and $2b_{\rm O}$,
$b_{\rm H}=2b_{\rm O}$. This is the so-called ``stoichiometric
mixture" where proportion between $b_{\rm H}$ and $2b_{\rm O}$ is
the same as in the ``product", ${\rm H}_2{\rm O}$.

\subsection{Examples of the thermodynamic tree}

\begin{figure}[t]
\centering
\includegraphics[width=0.62\textwidth]{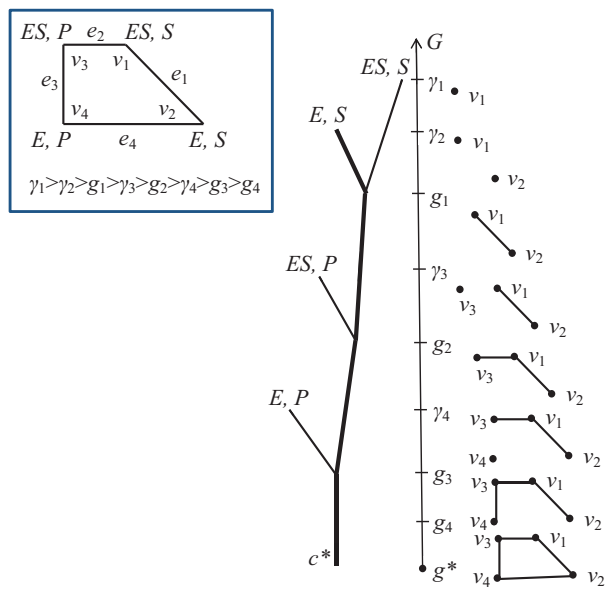}
\caption{\label{TreeTrap}The thermodynamic tree for the four--component
enzyme--substrate system $S$, $E$, $ES$ $P$ (Fig.~\ref{Trapecia}) with excess
of substrate: $b_S>b_E$ (case (a)). The vertices and edges are enumerated in
order of $\gamma_v$ and $g_e$ (starting from the greatest values). The order of
these numbers is represented in Fig. On the right, the graphs
$\widetilde{D_1}\setminus U_g$ are depicted. The solid bold line on the tree is
the thermodynamically admissible path from the initial state  $E,S$ (enzyme
plus substrate) to the equilibrium. There are leaves at all levels
$g=\gamma_i$. There are branching points at $g=g_{1,2,3}$ and no vertices at
$g=g_4$.}
\end{figure}

In this section, we present two example of the thermodynamic tree. First, let
us consider the trapezium (Fig.~\ref{Trapecia}a). Let us select the order of
numbers $\gamma_v$ and $g_e$ according to Fig.~\ref{TreeTrap}. The vertices and
edges are enumerated in order of $\gamma_v$ and $g_e$ (starting from the
greatest values). The tree is presented in Fig.~\ref{TreeTrap}. On the right,
the graphs $\widetilde{D_1}\setminus U_g$ are depicted for all intervals
$(a_{i-1},a_i]$. For $(\gamma_2,\gamma_1]$ it is just a vertex $v_1$. For
$(g_1,\gamma_2]$ it consists of two disjoint vertices, $v_1$ and $v_2$. For
$(\gamma_3,g_1]$ these two vertices are connected by an edge. On the interval
$(g_2, \gamma_3]$ the graph $\widetilde{D_1}\setminus U_g$ is an edge
$(v_1,v_2)$ and an isolated vertex $v_3$. On $(\gamma_4, g_2]$ all three
vertices $v_1$, $v_2$ and $v_3$ are connected by edges. For $(g_3,\gamma_4]$
the isolated vertex $v_4$ is added to the graph $\widetilde{D_1}\setminus U_g$.
For $g \leq g_3$ the graph $\widetilde{D_1}\setminus U_g$ includes all the
vertices and is connected.

\begin{figure}[t]
\centering
\includegraphics[width=0.58\textwidth]{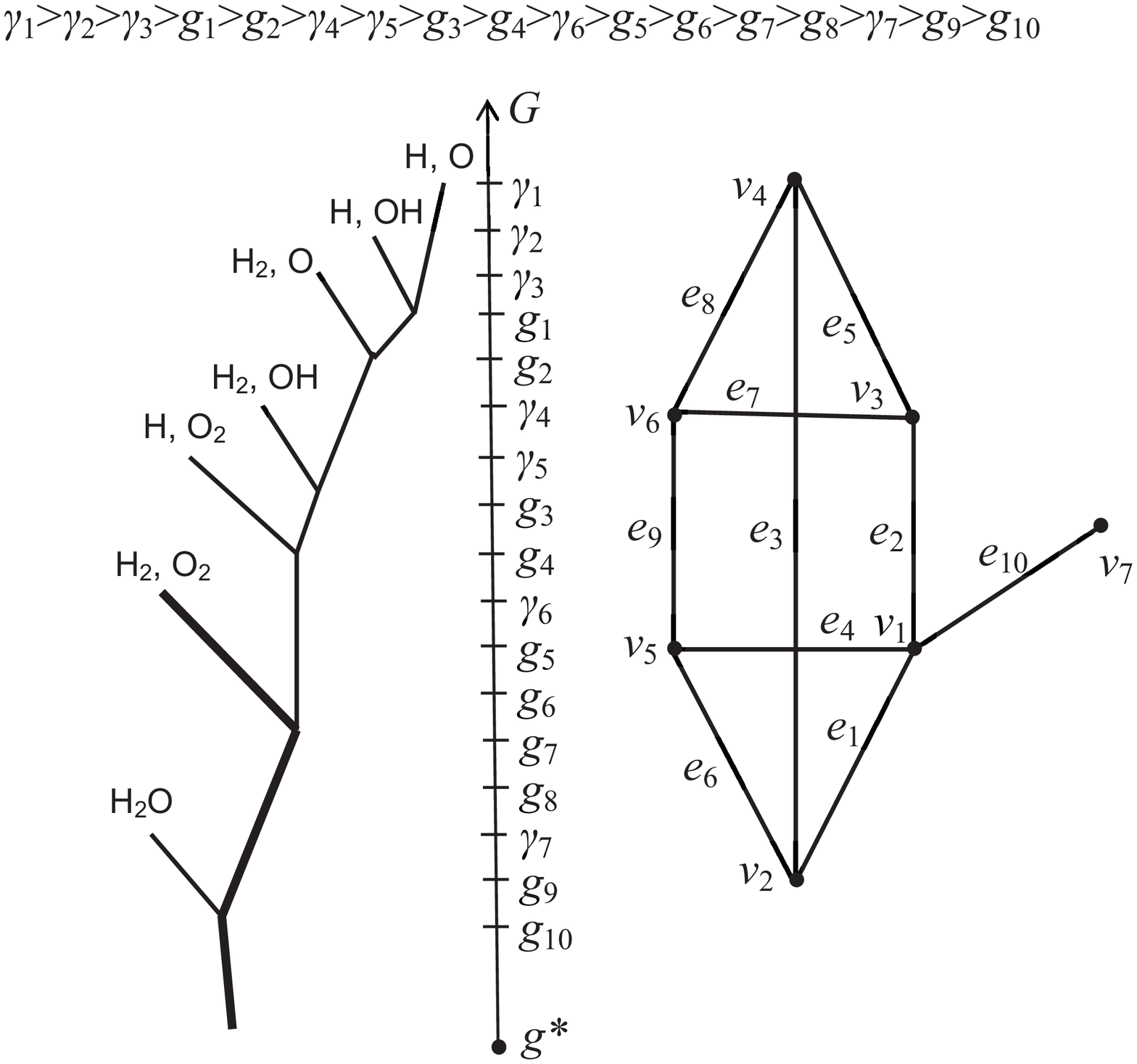}
\caption{\label{Tree6stoi} The thermodynamic tree for the six--component ${\rm
H}_2$--${\rm O}_2$ system, ${\rm H}_2$, ${\rm O}_2$, ${\rm H}$, ${\rm O}$,
${\rm H}_2{\rm O}$, ${\rm OH}$ with the stoichiometric hydrogen--oxygen ratio,
$b_{\rm H}=2 b_{\rm O}$ (Fig.~\ref{Paramet6}b). The order of numbers
$\gamma_i$, $g_j$ is presented in Fig. On the right, the graph
$\widetilde{D_1}\setminus U_g$ is represented for $g=g_{10}$. For $g\leq
g_{10}$, the graph $\widetilde{D_1}\setminus U_g$ includes all the vertices and
is connected. The solid bold line on the tree is the thermodynamically
admissible path from the initial state  ${\rm H}_2, {\rm O}_2$ to the
equilibrium. There are leaves at all levels $g=\gamma_i$. There are branching
points at $g=g_{1-4,7,10}$ and no vertices at $g=g_{5,6,8,9}$.}
\end{figure}

For the second example (Fig.~\ref{Tree6stoi}) we selected the six-component
system (Fig.~\ref{Paramet6}) with the stoichiometric hydrogen--oxygen ratio,
$b_{\rm H}=2 b_{\rm O}$. The selected order of numbers $\gamma_i$, $g_j$ is
presented in Fig.~\ref{Tree6stoi}.

\section{Conclusion}

We studied dynamical systems that obey a continuous strictly convex Lyapunov
function $G$ in a positively invariant convex polyhedron $D$. Convexity allows
us to transform $n$-dimensional problems about attainability and attainable
sets into an analysis of 1D continua and discrete objects.

We construct the tree (the Adelson-Velskii -- Kronrod -- Reeb tree
\cite{AdelKronrod1945,Kronrod1950,Reeb1946}) of the function $G$ in $D$ and
call this 1D continuum the {\em thermodynamic tree}.

The thermodynamic tree is a tool to solve the ``attainability problem": is
there a continuous path between two states, $x$ and $y$ along which the
conservation laws hold, the concentrations remain non-negative and the relevant
thermodynamic potential $G$ (Gibbs energy, for example) monotonically
decreases? This question arises often in non-equilibrium thermodynamics and
kinetics. The analysis of the admissible paths can be considered as a dynamical
analogue of the study of the steady states feasibility in chemical and
biochemical kinetics. In this recent study, the energy balance method, the
stoichiometric network theory, the entropy production analysis and the advanced
algorithms of convex geometry of polyhedral cones are used
\cite{BerdQuian2004,Qian2005}.

The obvious inequality, $G(x)\geq G(y)$ is necessary but not sufficient
condition for existence of an admissible path from $x$ to $y$. In 1D systems,
the space of states is an interval and the thermodynamic tree has two leaves
(the ends of the interval) and one root (the equilibrium). In such a system, a
spontaneous transition from a state $x$ to a state $y$ is allowed by
thermodynamics if $G(x)\geq G(y)$ and $x$ and $y$ are on the same side of the
equilibrium, i.e. they belong to the same branch of the thermodynamic tree.
This is just a well known rule: ``it is impossible to overstep the equilibrium
in one-dimensional systems".

The construction of the thermodynamic tree gives us the multidimensional
analogue of this rule. Let $\pi: D \to \mathcal{T}$ be the natural projection
of the balance polyhedron $D$ on the thermodynamic tree $\mathcal{T}$. A
spontaneous transition from a state $x$ to a state $y$ is allowed by
thermodynamics if and only if $\pi(y) \in [\pi(x),\pi(N^*)]$, where $N^*$ is
the equilibrium and $[\pi(x),\pi(N^*)]$  is the ordered segment.

In this paper, we developed methods for solving the following
problems:
\begin{enumerate}
\item{How to construct the thermodynamic tree $\mathcal{T}$?}
\item{How to solve the attainability problem?}
\item{How to describe the set of all states
attainable from a given initial state $x$?}
\end{enumerate}

For this purpose, we analyzed the cutting of a convex polyhedron by a convex
set and developed the algorithm for construction of the tree of level set
components of a convex function in a convex polyhedron. In this algorithm, the
restriction of $G$ onto the 1-skeleton of $D$ is used. This finite family of
convex functions of one variable includes all necessary information for
analysis of the tree of the level set component of the convex function $G$ of
many variables.

In high dimensions, some steps of our analysis become computationally
expensive. The most expensive operations are the convex hull (description of
the convex hull of a finite set by linear inequalities) and the double
description operations (description of the faces of a polyhedron given by a set
of linear inequalities). Therefore, in high dimensions some of the problem may
be modified, for example, instead of the explicit description of the convex
hull it is possible to use the algorithm for solution of a problem: does a
point belong to this convex hull \cite{PreparataShamos1985}. The computational
aspects of the discussed problems in higher dimensions deserve more attention
and the proper modifications of the problems should be elaborated. For example,
two following problems  need to be solved efficiently:
\begin{itemize}
\item{To find the maximal and the minimal value of any
    linear function $f$ in a class of thermodynamic
    equivalence;}
\item{To evaluate the maximum and the minimum of $\D G/\D
    t$ in any class of thermodynamic equivalence:
    $-\overline{\sigma}\leq \D G/\D t \leq
    -\underline{\sigma}\leq 0.$}
\end{itemize}

For any $w\in \mathcal{T}$, the solution of the first problem allows
us to find an interval of values of any linear function of state in
the corresponding class of thermodynamic equivalence. We can use the
results of Sec.~\ref{Sec:DescrIneq} to reformulate this problem as
the convex programming problem.

The second problem gives us the possibility to consider
dynamics of relaxation on $\mathcal{T}$. On each interval on
$\mathcal{T}$ we can write
\begin{equation}\label{TreeDyn}
-\overline{\sigma}(g)\leq\D g/\D t\leq -\underline{\sigma}(g)\leq 0\, ,
\end{equation}
where the functions $\overline{\sigma}(g),\underline{\sigma}(g)
\geq 0$ depend on the interval on $\mathcal{T}$.

This differential inequality (\ref{TreeDyn}) will be a tool for the study of
the dynamics of relaxation and may be considered as a reduced kinetic model
that substitutes dynamics on the $d$-dimensional balance polyhedron $D$ by
dynamics on the one-dimensional dendrite. The problem of the construction of
the reduced model (\ref{TreeDyn}) is closely related to the following problem
\cite{Villani}: ``Can one establish a lower bound on the entropy production, in
terms of how much the distribution function departs from thermodynamical
equilibrium?" In 1982, C. Cercignani \cite{Cercignani1982} proposed a simple
linear estimate for $\underline{\sigma}(g)$ for the Boltzmann equation
(Cercignani's conjecture). After that, these estimates were studied and
improved by many authors \cite{Desvillettes,Carlen,Toscani,Villani} and now the
state of art achieved for the Boltzmann equation gives us some hints how to
create the relaxation model (\ref{TreeDyn}) on the thermodynamic tree for the
general kinetic systems. This may be the next step in the study of the
thermodynamic trees.

\end{document}